\documentclass{article}
\usepackage{amssymb}
\usepackage{amsmath}
\usepackage{amsthm}
\usepackage{graphicx} 
\usepackage{color}
\usepackage[colorlinks]{hyperref}
\usepackage[a4paper]{geometry}
\usepackage{algorithm}
\usepackage{algorithmic}
\usepackage{booktabs}
\usepackage{float}

\newtheorem{theorem}{Theorem}[section]
\newtheorem{lemma}[theorem]{Lemma}
\newtheorem{proposition}[theorem]{Proposition}
\newtheorem{corollary}[theorem]{Corollary}

\theoremstyle{definition}

\newtheorem{example}{Example}[section]

\theoremstyle{remark}
\newtheorem{remark}{Remark}

\title{Embedding linear codes over $\mathbb{Z}_4$ into self-orthogonal codes}
\author{Junmin An\thanks{junmin0518@sogang.ac.kr, Department of Mathematics and Institute for Mathematical and Data Sciences, Sogang University, Seoul, Korea}, Jon-Lark Kim\thanks{jlkim@sogang.ac.kr, Department of Mathematics and Institute for Mathematical and Data Sciences, Sogang University, Seoul, Korea}, San Ling\thanks{lingsan@ntu.edu.sg, School of Physical and Mathematical Sciences, Nanyang Technological University, Singapore}~\thanks{ling.s@vinuni.edu.vn, VinUniversity, Vinhomes Ocean Park, Gia Lam, Hanoi, Vietnam.}}
\date{}

\begin{document}
\maketitle
\begin{abstract}
The purpose of this paper is to investigate the self-orthogonal embedding problem for linear codes over $\mathbb{Z}_4$. We propose several tight bounds on the length of the shortest self-orthogonal embedding over $\mathbb{Z}_4$, and determine the exact shortest self-orthogonal embedding length under specific conditions. As an example satisfying these conditions, we establish the exact length of the shortest self-orthogonal embedding for the quaternary Preparata codes. Furthermore, to establish these results, we completely classify the exact length of the shortest doubly even self-orthogonal embedding for binary linear codes in every possible case. Finally, when the shortest self-orthogonal embedding length of a given free code over $\mathbb{Z}_4$ is equal to the shortest doubly even self-orthogonal embedding length of its residue code, we present an algorithm to construct all possible shortest self-orthogonal embeddings. With our algorithm, we found twelve linear codes over $\mathbb{Z}_4$ whose minimum Lee distances are higher than those of the $\mathbb{Z}_4$-linear codes in Aydin’s database.
\end{abstract}

\section{Introduction}
Constructing good new codes from a given linear code has become an important problem. As one such approach, the embedding method which creates a new code by adding columns to an existing one has emerged as a novel method for code construction. The central problem is twofold: first, determining the minimum number of columns required to embed a given code so that the resulting code is self-orthogonal, LCD, or has a prescribed hull dimension; and second, how to explicitly construct or obtain such embedded codes with the highest minimum distance. This problem was first investigated in~\cite{KKL-2021}. Specifically, the authors devised an algorithm to obtain binary self-orthogonal codes by adding the minimum number of columns, and applied it to binary linear codes of dimensions three and four to construct optimal codes. Subsequently, Kim and Choi~\cite{KC-2022} devised a new algorithm using a self-orthogonality matrix, and extended the previous results up to dimension 6. An et al.~\cite{AKKLW-arXiv} utilized the hull of a linear code to completely determine the minimum number of columns required to embed a given binary linear code into a self-orthogonal code for all cases, and constructed optimal codes of dimension $8$. Since it was shown in \cite{CMTQP-2018-2} that any linear code over a field of order 5 or greater is equivalent to an LCD code, An et al. \cite{AHKL-arXiv} extended the previous shortest self-orthogonal embeddings to address the shortest LCD embedding problem, which determines the minimum number of columns required to embed a given linear code into an LCD code over $\mathbb{F}_2$, $\mathbb{F}_3$ and $\mathbb{F}_4$. Most recently, Wang and Luo \cite{WL-arXiv} further generalized these existing results by investigating the problem of embedding a given linear code over any finite field to achieve a prescribed hull dimension.

We aim to shift the focus to codes over rings, specifically to study the self-orthogonal embedding problem over $\mathbb{Z}_4$. The intensive study of codes over $\mathbb{Z}_4$ was primarily motivated by the observation that certain optimal binary nonlinear codes can be represented as the images of linear codes over $\mathbb{Z}_4$ under a specific map. Forney et al.~\cite{FST-1992} showed that the Nordstrom-Robinson code, which is an optimal binary nonlinear code, is the image of the octacode over $\mathbb{Z}_4$. Hammons et al.~\cite{HKCSS-1994} later showed that infinite families of very good binary nonlinear codes, including the Kerdock, Preparata, and Goethals codes, similarly arise as the images of linear codes over $\mathbb{Z}_4$. Self-dual codes over $\mathbb{Z}_4$ have been widely studied~\cite{BSBM-1997,CS-1993,DHS-2001,FGPL-1998,G-1996,H-2010,H-2012,HK-2000,M-2024,R-1999,R-2000}, with many classical results over fields including invariant theory, weight enumerators, extremal self-dual codes and shadow theory being extensively generalized to the $\mathbb{Z}_4$ setting. Similarly, well-known families of codes, such as cyclic codes and quadratic residue (QR) codes, have also been generalized to $\mathbb{Z}_4$~\cite{BSC-1995,CMKH-1996,CS-1995,PQ-1996}. Therefore, it is a natural question to ask how to extend the self-orthogonal embedding problem from finite fields to $\mathbb{Z}_4$.

Another key motivation for our research is its connection to the problem of expanding self-orthogonal codes to self-dual codes. Recently, Kim et al.~\cite{KLL-2026} showed that for any Euclidean and Hermitian self-orthogonal code over a finite field, the dimension can always be increased by one while preserving the self-orthogonality property. Through this process, they demonstrated that a self-dual code can be constructed from a self-orthogonal code. Similarly, Shi and Tao~\cite{ST-2026} showed that a self-orthogonal code over $\mathbb{Z}_4$ can also be expanded into a self-dual code and through this approach, they obtained new optimal self-dual codes over $\mathbb{Z}_4$. Therefore, if a given linear code over $\mathbb{Z}_4$ can be embedded into a good self-orthogonal code, it can be combined with the expanding method to construct a good self-dual code over $\mathbb{Z}_4$.

In this paper, we discover that the key to solving the self-orthogonal embedding problem over $\mathbb{Z}_4$ lies in the shortest doubly even self-orthogonal embedding of its binary residue code. Thus, before addressing the self-orthogonal embedding problem over $\mathbb{Z}_4$, we investigate the shortest doubly even self-orthogonal embedding problem for binary codes, which is the problem of finding the minimum number of columns required to embed a given linear code so that the resulting code is not only self-orthogonal but also doubly even. First, we show that the number of additional columns required for the shortest doubly even self-orthogonal embedding exceeds that of the shortest self-orthogonal embedding by at most two. Furthermore, we completely classify the exact number of required columns for every possible case. Next, regarding the self-orthogonal embedding over $\mathbb{Z}_4$, we propose a bound on the number of columns required for the shortest self-orthogonal embedding. Specifically, given a linear code over $\mathbb{Z}_4$ with a generator matrix $G$, we show that the minimum number of columns $m$ required to embed it into a self-orthogonal code satisfies $a+b\le m\le 3a+2b$, where $4^a2^b$ is the type of the code generated by $GG^T$. Using this, we derive two other tight bounds with examples which demonstrate that these bounds are tight. Our main result is that if a code over $\mathbb{Z}_4$ satisfies certain conditions, such as its residue code being LCD, then its shortest self-orthogonal embedding can be obtained from the shortest doubly even self-orthogonal embedding of its residue code. Since we have already completely determined the exact number of columns required for the doubly even case, this result allows us to find the exact value for the corresponding code over $\mathbb{Z}_4$. By applying this result, we exactly determine the length of the shortest self-orthogonal embedding for the quaternary Preparata codes. Finally, when the length of a shortest self-orthogonal embedding of a given free code over $\mathbb{Z}_4$ is equal to the length of a shortest self-orthogonal embedding of its residue code, we present an algorithm to construct all possible shortest self-orthogonal embeddings. Using this algorithm, we provide several construction examples along with their resulting parameters.

This paper is organized as follows. In Section 2, we provide preliminaries on binary codes and codes over $\mathbb{Z}_4$, and introduce previous results on shortest self-orthogonal embeddings. In Section 3, we present our results on shortest doubly even self-orthogonal embeddings. By dividing the codes into even and odd cases, we classify the exact length of the shortest doubly even self-orthogonal embedding for each specific case. In Section 4, we present several tight bounds on the length of the shortest self-orthogonal embedding over $\mathbb{Z}_4$. Furthermore, under a specific condition, we determine the exact length by using the residue code of the corresponding code. In Section 5, we propose an algorithm for constructing shortest self-orthogonal embeddings and provide explicit examples using this algorithm. We conclude our paper in Section 6.

\section{Preliminaries}
For a detailed background on coding theory, we refer the reader to~\cite{book-HP,book-MS}. For $\mathbb{Z}_4$ codes, we refer to~\cite{book-W}.

Let $\mathbb{F}_2$ be the field with two elements. For $n\ge 1$, a nonempty subset $\mathcal{C}$ of $\mathbb{F}_2^n$ is called a \textit{code} over $\mathbb{F}_2$ with length $n$, and an element of $\mathcal{C}$ is called a \textit{codeword}. If a code $\mathcal{C}$ is a $k$-dimensional subspace of $\mathbb{F}_2^n$, then it is called a \textit{linear code} over $\mathbb{F}_2$, denoted by an $[n, k]$ code. Otherwise, $\mathcal{C}$ is called a \textit{nonlinear code}. For an $[n, k]$ linear code $\mathcal{C}$ over $\mathbb{F}_2$, a \textit{generator matrix} $G$ for $\mathcal{C}$ is a $k\times n$ matrix over $\mathbb{F}_2$ whose rows form a basis of $\mathcal{C}$ over $\mathbb{F}_2$. Two codes $\mathcal{C}_1$ and $\mathcal{C}_2$ over $\mathbb{F}_2$ are said to be \textit{equivalent} if there is a permutation of columns $\sigma$ such that $\sigma\mathcal{C}_1=\mathcal{C}_2$.

Similarly, a \textit{code} $\mathcal{C}$ over $\mathbb{Z}_4$ is a nonempty subset of $\mathbb{Z}_4^n$ for $n\ge 1$, and is \textit{linear} if $\mathcal{C}$ is a $\mathbb{Z}_4$-submodule of $\mathbb{Z}_4^n$. A $\mathbb{Z}_4$ linear code $\mathcal{C}$ of length $n$ contains a set of $k_1+k_2$ codewords $\mathbf{c}_1, \ldots, \mathbf{c}_{k_1}\in\mathbb{Z}_4^n\backslash2\mathbb{Z}_4^n$ and $\mathbf{c}_{k_1+1}, \ldots, \mathbf{c}_{k_1+k_2}\in 2\mathbb{Z}_4^n$ such that every codeword in $\mathcal{C}$ is uniquely expressible in the form
\[
\sum_{i=1}^{k_1}a_i\mathbf{c}_i+\sum_{j=k_1+1}^{k_1+k_2}b_j\mathbf{c}_j
\]
where $a_i\in\mathbb{Z}_4$ and $b_j\in\mathbb{Z}_2$. The matrix whose rows are $\mathbf{c}_i$ for $1\le i\le k_1+k_2$ is called a \textit{generator matrix} for $\mathcal{C}$, and $\mathcal{C}$ is said to be of type $4^{k_1}2^{k_2}$. Here, $k_1$ is called the \textit{free rank} of $\mathcal{C}$. If $k_2=0$, then $\mathcal{C}$ is called a \textit{free code}, and $k_1$ is called the \textit{rank} of the free code $\mathcal{C}$. Two codes $\mathcal{C}_1$ and $\mathcal{C}_2$ over $\mathbb{Z}_4$ are called \textit{equivalent} if there is a monomial matrix $M$ which has exactly one nonzero entry from $\{1, 3\}$ in each row and column, such that $\mathcal{C}_1M=\mathcal{C}_2$. Throughout this paper, the term \textit{code} refers to a linear code unless otherwise stated.

For a code $\mathcal{C}$ over $\mathbb{F}_2$ or $\mathbb{Z}_4$, we denote a generator matrix of $\mathcal{C}$ by $G(\mathcal{C})$. Conversely, for a matrix $A$ over $\mathbb{F}_2$ or $\mathbb{Z}_4$, we denote the code generated by the rows of $A$ by $\langle A\rangle$.

For a code $\mathcal{C}$ over $\mathbb{F}_2$ (resp. $\mathbb{Z}_4)$ with length $n$, the \textit{dual} $\mathcal{C}^\perp$ of the code $\mathcal{C}$ is defined as the set $\{\mathbf{x}\in\mathbb{F}_2^n~ (\text{resp. }\mathbb{Z}_4)~|~\mathbf{x}\cdot \mathbf{c}=0\text{ for every }\mathbf{c}\in\mathcal{C}\}$ where $\cdot$ is the standard dot product. For an $[n, k]$ code $\mathcal{C}$ over $\mathbb{F}_2$, $\dim\mathcal{C}^\perp=n-k$. For a code $\mathcal{C}$ over $\mathbb{Z}_4$ with length $n$ and type $4^{k_1}2^{k_2}$, the type of $\mathcal{C}^\perp$ is $4^{n-k_1-k_2}2^{k_2}$. If $\mathcal{C}\subseteq\mathcal{C}^\perp$, then $\mathcal{C}$ is called a \textit{self-orthogonal code}, and if $\mathcal{C}=\mathcal{C}^\perp$, then it is called a \textit{self-dual code}. The \textit{hull} of $\mathcal{C}$ is defined as $\mathcal{C}\cap\mathcal{C}^\perp$, denoted by ${\rm Hull}(\mathcal{C})$. If ${\rm Hull}(\mathcal{C})=\{\mathbf{0}\}$, then $\mathcal{C}$ is called an \textit{LCD code}. The dimension of the hull of a code $\mathcal{C}$ over $\mathbb{F}_2$ can be obtained using the following theorem.
\begin{theorem}[\cite{LZ-2019}]
    Let $\mathcal{C}$ be an $[n, k]$ code over $\mathbb{F}_2$ with $\dim {\rm Hull}(\mathcal{C})=\ell$. Then
    \[
    \ell=k-\text{rank}(GG^T)
    \]
    where $G$ is a generator matrix for $\mathcal{C}$.
\end{theorem}
Hence, $\mathcal{C}$ is self-orthogonal if and only if $GG^T=\mathcal{O}$, and $\mathcal{C}$ is LCD if and only if $GG^T$ is invertible.

The \textit{Hamming weight} of a vector $\mathbf{x}\in\mathbb{F}_2^n$ (resp. $\mathbb{Z}_4^n$), denoted by ${\rm wt}(\mathbf{x})$ is the number of nonzero coordinates in $\mathbf{x}$. For two vectors $\mathbf{x}, \mathbf{y}$ in $\mathbb{F}_2^n$ (resp. $\mathbb{Z}_4^n$), the \textit{Hamming distance} between $\mathbf{x}$ and $\mathbf{y}$ is given as $d(\mathbf{x}, \mathbf{y})={\rm wt}(\mathbf{x}-\mathbf{y})$. The \textit{minimum (Hamming) distance} of a code $\mathcal{C}$, denoted by $d(\mathcal{C})$, is the minimum of the (Hamming) distances between any two distinct codewords in $\mathcal{C}$. For a code $\mathcal{C}$ over $\mathbb{F}_2$, if every codeword of a code $\mathcal{C}$ has even (Hamming) weight, then $\mathcal{C}$ is called an \textit{even code}. Otherwise, it is called an \textit{odd-like code}. For an even code $\mathcal{C}$, if every codeword of $\mathcal{C}$ has weight divisible by four, then it is called a \textit{doubly even code}. For a vector $\mathbf{x}\in\mathbb{Z}_4^n$, let $n_a(\mathbf{x})$ denote the number of components of $\mathbf{x}$ equal to $a$ for $a\in\mathbb{Z}_4$. The \textit{Lee weight} of $\mathbf{x}$ is ${\rm wt}_L(\mathbf{x})=n_1(\mathbf{x})+2n_2(\mathbf{x})+n_3(\mathbf{x})$. For two vectors $\mathbf{x}, \mathbf{y}\in\mathbb{Z}_4^n$, the \textit{Lee distance} between $\mathbf{x}$ and $\mathbf{y}$ is defined as $d_L(\mathbf{x}, \mathbf{y})={\rm wt}_L(\mathbf{x}-\mathbf{y})$.

\begin{theorem}[\cite{CMTQ-2018-1}]\label{LCDbasis}
Let $\mathcal{C}$ be an even code over $\mathbb{F}_2$. Then $\mathcal{C}$ is LCD if and only if the dimension $k$ of $\mathcal{C}$ is even and there exists a basis $\mathbf{c}_1, \mathbf{c}_1', \mathbf{c}_2, \mathbf{c}_2',\ldots, \mathbf{c}_{k/2}, \mathbf{c}_{k/2}'$ of $\mathcal{C}$ such that for any $i, j\in\{1, 2, \ldots, k/2\}$, the following holds:
\begin{enumerate}
    \item[(i)] $\mathbf{c}_i\cdot\mathbf{c}_j=\mathbf{c}_i'\cdot\mathbf{c}_j'=0$,
    \item[(ii)] $\mathbf{c}_i\cdot\mathbf{c}_j'=0$ if $i\ne j$, and
    \item[(iii)] $\mathbf{c}_i\cdot\mathbf{c}_i'=1$.
\end{enumerate}
\end{theorem}

For a code $\mathcal{C}$ over $\mathbb{Z}_4$, the \textit{residue code} of $\mathcal{C}$, denoted by $\text{Res}(\mathcal{C})$, is a code over $\mathbb{F}_2$ defined as $\text{Res}(\mathcal{C})=\mathcal{C}\pmod 2$. The \textit{torsion code} of $\mathcal{C}$, denoted by $\text{Tor}(\mathcal{C})$ is defined as $\text{Tor}(\mathcal{C})=\{\mathbf{x}\in\mathbb{F}_2^n~|~2\mathbf{x}\in\mathcal{C}\}$. It is known that $\text{Res}(\mathcal{C})^\perp=\text{Tor}(\mathcal{C}^\perp)$ and $\text{Tor}(\mathcal{C})^\perp=\text{Res}(\mathcal{C}^\perp)$.

Next, we introduce the notion of a shortest self-orthogonal embedding of a linear code $\mathcal{C}$ over $\mathbb{F}_2$ and some related results. Detailed information on shortest self-orthogonal embeddings can be found in~\cite{AKKLW-arXiv}.

For an $[n, k]$ code $\mathcal{C}$ over $\mathbb{F}_2$, a \textit{self-orthogonal embedding} $\mathcal{C}'$ of $\mathcal{C}$ is an $[n', k]$ self-orthogonal code over $\mathbb{F}_2$ such that $n'\ge n$ and $\mathcal{C}$ is a punctured code of $\mathcal{C}'$. A \textit{shortest self-orthogonal embedding} of $\mathcal{C}$ is a self-orthogonal embedding of $\mathcal{C}$ with the smallest length among all self-orthogonal embeddings of $\mathcal{C}$. We denote the length of a shortest self-orthogonal embedding of $\mathcal{C}$ by $n_s(\mathcal{C})$.

\begin{theorem}[\cite{AKKLW-arXiv}]\label{shortestSOlcd}
    Let $\mathcal{C}$ be a code over $\mathbb{F}_2$ with generator matrix
    \[
     G=\begin{bmatrix}
        G({\rm Hull}(\mathcal{C}))\\A
    \end{bmatrix}.
    \]
    Then $n_s(\mathcal{C})=n_s(\langle A\rangle)$.
\end{theorem}

\begin{theorem}[\cite{AKKLW-arXiv}]\label{shortestSO}
    Let $\mathcal{C}$ be an $[n, k]$ code over $\mathbb{F}_2$ with $\dim{\rm{Hull}}(\mathcal{C})=\ell$. If $\mathcal{C}$ is an even code, then $n_s(\mathcal{C})=n+k-\ell+1$. Otherwise, if $\mathcal{C}$ is an odd-like code, then $n_s(\mathcal{C})=n+k-\ell$.
\end{theorem}

\begin{theorem}[\cite{AKKLW-arXiv}]\label{shortestSOgen}
    Let $\mathcal{C}'$ be a shortest self-orthogonal embedding of a code $\mathcal{C}$ over $\mathbb{F}_2$. Then $\mathcal{C}'$ has a generator matrix
    \[
    G'=\begin{bmatrix}
        G({\rm Hull}(\mathcal{C})) & \mathcal{O}\\A & B
    \end{bmatrix}
    \]
    where $\mathcal{O}$ is the zero matrix and \[
    G=\begin{bmatrix}
        G({\rm Hull}(\mathcal{C}))\\A
    \end{bmatrix}
    \]
    is a generator matrix of $\mathcal{C}$.
\end{theorem}

Let $V$ be a finite dimensional vector space over $\mathbb{F}_2$. A function $q:V\to \mathbb{F}_2$ is called a \textit{quadratic form} if there exists a bilinear map $B:V\times V\to \mathbb{F}_2$ such that $q(x+y)+q(x)+q(y)=B(x, y)$ for every $x, y\in V$. If the bilinear map $B$ is nondegenerate, then $q$ is called \textit{nonsingular}. A \textit{quadratic space} $(V, q_V)$ is called \textit{nonsingular} if $q_V$ is a nonsingular quadratic form.

Let $V$ and $W$ be finite dimensional vector spaces over $\mathbb{F}_2$ with quadratic forms $q_V:V\to \mathbb{F}_2$ and $q_W:W\to \mathbb{F}_2$, respectively. Two quadratic spaces $(V, q_V)$ and $(W, q_W)$ are called \textit{isometric} if there is a linear isomorphism $T:V\to W$ such that $q_V(x)=q_W(T(x))$ for all $x\in V$.

Let $V$ be a vector space over $\mathbb{F}_2$ and let $q$ be a quadratic form on $V$. It is known that if $q$ is nonsingular, then one of the values $0$ or $1$ occurs more frequently than the other among the values of $q(x)$. The \textit{Arf invariant} ${\rm{Arf}}(q)\in\mathbb{F}_2$ of such a form $q$ is defined as this most frequently taken value as $x$ varies over $V$.

\begin{theorem}[\cite{A-1941}]\label{Arfiso}
Let $V$ and $W$ be finite dimensional vector spaces over $\mathbb{F}_2$ with quadratic forms $q_V:V\to \mathbb{F}_2$ and $q_W:W\to \mathbb{F}_2$ such that both $q_V$ and $q_W$ are nonsingular. Then $(V, q_V)$ and $(W, q_W)$ are isometric if and only if $\dim V=\dim W$ and ${\rm{Arf}}(q_V)={\rm{Arf}}(q_W)$.
\end{theorem}
As a consequence, every $k$-dimensional nonsingular quadratic space is classified into exactly two types up to isometry, depending on whether its Arf invariant is zero or one.

\section{Shortest doubly even self-orthogonal embeddings of binary codes}
For a linear code $\mathcal{C}$ over $\mathbb{F}_2$, we define $\mathcal{C}'$ to be a \textit{doubly even self-orthogonal embedding} of $\mathcal{C}$ if $\mathcal{C}'$ is a self-orthogonal embedding of $\mathcal{C}$ such that $\mathcal{C}'$ is doubly even. A shortest doubly even self-orthogonal embedding of $\mathcal{C}$ is a doubly even self-orthogonal embedding of $\mathcal{C}$ with the minimum possible length. We denote the length of a shortest doubly even self-orthogonal embedding of $\mathcal{C}$ by $n_{de}(\mathcal{C})$.

\begin{proposition}
    Let $\mathcal{C}$ be an $[n, k]$ code over $\mathbb{F}_2$. Then
    \[
    n_s(\mathcal{C})\le n_{de}(\mathcal{C})\le n_s(\mathcal{C})+2.
    \]
\end{proposition}
\begin{proof}
    If there is a shortest self-orthogonal embedding of $\mathcal{C}$ which is doubly even, then $n_s(\mathcal{C})=n_{de}(\mathcal{C})$. If not, then take a shortest self-orthogonal embedding $\mathcal{C}'$ of $\mathcal{C}$ with generator matrix $G'$. For each row $g$ of $G'$, if ${\rm{wt}}(g)$ is not divisible by $4$, then concatenate $11$ to $g$. On the other hand, if ${\rm{wt}}(g)$ is divisible by $4$, then concatenate $00$ to $g$. The resulting matrix $G''$ generates a doubly even self-orthogonal code of length $n_s(\mathcal{C})+2$, which implies $n_{de}(\mathcal{C})\le n_s(\mathcal{C})+2$.
\end{proof}

\begin{proposition}\label{deequals}
    Let $\mathcal{C}$ be an $[n, k]$ code over $\mathbb{F}_2$ with generator matrix
    \[
    G=\begin{bmatrix}
        G(\mathrm{Hull}(\mathcal{C}))\\A
    \end{bmatrix}.
    \]
    Then $n_{de}(\mathcal{C})=n_s(\mathcal{C})$ if and only if $\mathrm{Hull}(\mathcal{C})$ is doubly even and there is a shortest self-orthogonal embedding of $\langle A\rangle$ which is doubly even.
\end{proposition}
\begin{proof}
    Suppose that $n_{de}(\mathcal{C})=n_s(\mathcal{C})$. Then there is a shortest self-orthogonal embedding $\mathcal{C}'$ of $\mathcal{C}$ which is doubly even. By Theorem~\ref{shortestSOgen}, $\mathcal{C}'$ has a generator matrix
    \[
    G'=\begin{bmatrix}G({\rm Hull}(\mathcal{C})) & \mathcal{O}\\A & B \end{bmatrix}
    \]
    for some matrix $B$. Since $\mathcal{C}'$ is doubly even, ${\rm{Hull}}(\mathcal{C})$ and the code generated by $[A~|~B]$ are also doubly even. Moreover, $[A~|~B]$ generates a self-orthogonal code, and by Theorem~\ref{shortestSOlcd}, it is a shortest self-orthogonal embedding of $\langle A\rangle$. Hence, $[A~|~B]$ generates a shortest doubly even self-orthogonal embedding of $\langle A\rangle$.

    Conversely, suppose that $\langle A\rangle$ has a shortest self-orthogonal embedding with generator matrix $[A~|~D]$ which is doubly even. If ${\rm{Hull}}(\mathcal{C})$ is doubly even, then
    \[
    G''=\begin{bmatrix}G({\rm Hull}(\mathcal{C})) & \mathcal{O}\\A & D \end{bmatrix}
    \]
    generates a shortest self-orthogonal embedding of $\mathcal{C}$ which is doubly even. Thus $n_{de}(\mathcal{C})=n_s(\mathcal{C})$.
\end{proof}

\begin{proposition}
    Let $\mathcal{C}$ be a self-orthogonal code over $\mathbb{F}_2$, which is not doubly even. Then $n_{de}(\mathcal{C})=n_s(\mathcal{C})+2$.
\end{proposition}
\begin{proof}
    Choose a codeword $\mathbf{x}\in\mathcal{C}$ such that ${\rm wt}(\mathbf{x})\equiv 2\pmod 4$. If only one coordinate is appended, then the resulting vector $\mathbf{x}'$ has weight either ${\rm wt}(\mathbf{x})$ or ${\rm wt}(\mathbf{x})+1$. Thus, it cannot be doubly even.
\end{proof}

\begin{proposition}\label{doublyevengen}
Let $\mathcal{C}$ be an $[n, k]$ code over $\mathbb{F}_2$. If ${\rm Hull}(\mathcal{C})$ is doubly even, then there is a shortest doubly even self-orthogonal embedding $\mathcal{C}'$ of $\mathcal{C}$ with generator matrix
\[
    G'=\begin{bmatrix}
        G({\rm Hull}(\mathcal{C})) & \mathcal{O}\\A & B
    \end{bmatrix},
\]
where
\[
G=\begin{bmatrix}
        G({\rm Hull}(\mathcal{C}))\\A
    \end{bmatrix}
\]
is a generator matrix of $\mathcal{C}$. 
\end{proposition}
\begin{proof}
    Let
    \[
    G_0=\begin{bmatrix}
        G({\rm Hull}(\mathcal{C})) & B_H\\A & B_A
    \end{bmatrix}
    \]
    be a generator matrix of a shortest doubly even self-orthogonal embedding $\mathcal{C}_0$ of $\mathcal{C}$. Then the matrix
    \[
    G'=\begin{bmatrix}
        G({\rm Hull}(\mathcal{C})) & \mathcal{O}\\A & B_A
    \end{bmatrix}
    \]
    generates a shortest doubly even self-orthogonal embedding of $\mathcal{C}$.
\end{proof}

\begin{lemma}\label{doublyevencodeword}
    For $r\ge 2$, let $E_{r}$ be the binary even weight code of length $r$, which is given as
    \[
    E_r=\{\mathbf{x}\in\mathbb{F}_2^r~|~{\rm wt}(\mathbf{x})\mbox{ is even}\}.
    \]
    Then the number of doubly even codewords in $E_{r}$ is
    \[
    \begin{cases}
        2^{r-2}+(-1)^{\lfloor\frac{r+1}{4}\rfloor}2^{\frac{r-3}{2}} & \text{if }r\text{ is odd},\\
        2^{r-2}+(-1)^{\frac{r}{4}}2^{\frac{r}{2}-1} & \text{if }r\equiv 0\pmod 4,\\
        2^{r-2} & \text{if }r\equiv 2\pmod 4.
    \end{cases}
    \]
\end{lemma}
\begin{proof}
    It is well-known that the weight enumerator of $E_r$ is given as
    \[
    W_{E_r}(z)=\frac{1}{2}((1+z)^r+(1-z)^r).
    \]
    Let $k_0$ and $k_1$ be the numbers of doubly even codewords and singly even codewords in $E_{r}$, respectively. Then $k_0+k_1=|E_r|=2^{r-1}$, and
    \[
    k_0-k_1=\sum_{\mathbf{x}\in E_r} i^{{\rm{wt}}(\mathbf{x})} = W_{E_r}(i),
    \]
    where $i$ denotes the imaginary unit satisfying $i^2=-1$. Since $1+i=\sqrt{2}e^{\pi i/4}$ and $1-i=\sqrt{2}e^{-\pi i/4}$, we obtain
    \begin{align*}
    W_{E_r}(i)&=\frac{1}{2}((1+i)^r+(1-i)^r)\\&=\frac{2^{r/2}}{2}(e^{r\pi i/4}+e^{-r\pi i/4})\\&=2^{r/2}\cos\left({\frac{r\pi}{4}}\right).
    \end{align*}
    Therefore,
    \[
    k_0=\frac{1}{2}((k_0+k_1)+(k_0-k_1))=2^{r-2}+2^{\frac{r}{2}-1}\cos\left(\frac{r\pi}{4}\right).
    \]
    This completes the proof.
\end{proof}

\begin{remark}\label{evenArf}
For each positive integer $n$, define a map
$f:E_n\to\mathbb{F}_2$ by
\[
f(\mathbf{x})=\frac{\operatorname{wt}(\mathbf{x})}{2}\pmod 2
\]
for every $\mathbf{x}\in E_n$. Throughout this section, for any positive integer $n$, we use the same notation $f$ for these maps, since the domain is determined by the length of the vector under consideration. For any
$\mathbf{x},\mathbf{y}\in E_n$, we have
\begin{align*}
    f(\mathbf{x}+\mathbf{y})&=\frac{{\rm{wt}}(\mathbf{x+\mathbf{y}})}{2}\pmod 2\\&=\frac{1}{2}({\rm{wt}}(\mathbf{x})+{\rm{wt}}(\mathbf{y})+2\mathbf{x}\cdot \mathbf{y})\pmod 2\\&=f(\mathbf{x})+f(\mathbf{y})+\mathbf{x}\cdot\mathbf{y}.
\end{align*}
Thus, $f$ is a quadratic form on $E_n$. Let $\mathcal{C}$ be an even code of length $n$ over $\mathbb{F}_2$. Then $f|_{\mathcal{C}}$ is a quadratic form on $\mathcal{C}$. If $\mathcal{C}$ is an LCD code, then the standard dot product on $\mathcal{C}$ is nondegenerate, and therefore $(\mathcal{C}, f|_\mathcal{C})$ is nonsingular.

Assume that $(\mathcal{C}, f|_\mathcal{C})$ is nonsingular. Note that if $\mathbf{x}\in\mathcal{C}$ is doubly even, then $f(\mathbf{x})=0$, and if not, then $f(\mathbf{x})=1$. Therefore, counting the number of doubly even codewords is equivalent to counting the number of codewords $\mathbf{x}\in\mathcal{C}$ such that $f(\mathbf{x})=0$. Consequently, if the number of doubly even codewords is greater than the number of singly even codewords, then ${\rm Arf}(f|_\mathcal{C})=0$. Otherwise, ${\rm Arf}(f|_\mathcal{C})=1$. Therefore, by Theorem~\ref{Arfiso}, the isometry type of an even LCD code $\mathcal{C}$ with respect to the quadratic form $f|_\mathcal{C}$ is determined by the number of doubly even codewords.
\end{remark}

\begin{lemma}\label{LCDeveniso}
    Let $\mathcal{C}$ be an $[n, k]$ even code over $\mathbb{F}_2$ and let $E_r$ be the even weight code over $\mathbb{F}_2$ with length $r$. Then there is a doubly even self-orthogonal embedding of $\mathcal{C}$ with length $n+r$ if and only if there exists a linear map $T:\mathcal C\to E_r$ such that $f(T(\mathbf c))=f(\mathbf c)$ for all $\mathbf c\in\mathcal C$. Moreover, if $\mathcal{C}$ is LCD, then such a map $T$ is injective.
\end{lemma}
\begin{proof}
    Assume that there is a doubly even self-orthogonal embedding $\mathcal{C}'$ of $\mathcal{C}$ with length $n+r$. Let $G'=[G~|~B]$ be a generator matrix of $\mathcal{C}'$ where $G$ is a generator matrix of $\mathcal{C}$. Note that for a codeword $\mathbf{c}\in\mathcal{C}$, $\mathbf{c}=\mathbf{x}G$ for some unique $\mathbf{x}\in\mathbb{F}_2^{k}$. Define a linear map $T:\mathcal{C}\to \langle B\rangle$ as $T(\mathbf{x}G)=\mathbf{x}B$ for all $\mathbf{x}\in\mathbb{F}_2^{k}$. Since $(\mathbf{x}G, \mathbf{x}B)$ is doubly even and $\mathbf{x}G$ is even, it follows that $\mathbf{x}B$ is even, that is, $\mathbf{x}B\in E_r$. Also, we have
    \[
    {\rm{wt}}(\mathbf{x}G)+{\rm{wt}}(\mathbf{x}B)\equiv 0\pmod 4,
    \]
    that is,
    \[
    f(\mathbf{x}G)=f(\mathbf{x}B)=f(T(\mathbf{x}G)).
    \]

    On the other hand, assume that there is a linear map $T':\mathcal{C}\to E_r$ such that $f(T'(\mathbf c))=f(\mathbf c)$ for all $\mathbf c\in\mathcal C$. For each row $g_i$ of a generator matrix $G$ of $\mathcal{C}$, define $a_i=T'(g_i)$ for $1\le i\le k$. Let $A$ be the matrix whose rows are $a_1, \ldots, a_{k}$. We claim that $G''=[G~|~A]$ generates a doubly even self-orthogonal code $\mathcal{C}''$. Since $f(\mathbf{x})=f(T'(\mathbf{x}))$, we have
    \[
    {\rm{wt}}(\mathbf{x})+{\rm{wt}}(T'(\mathbf{x}))\equiv 0\pmod 4
    \]
    for every $\mathbf{x}\in\mathcal{C}$. Thus, $\mathcal{C}''$ is doubly even and hence self-orthogonal.    

    Finally, we prove the injectivity of $T$. Let $\mathcal{C}$ be an LCD code. Suppose $T(\mathbf{u})=\mathbf{0}$ for some $\mathbf{u}\in\mathcal{C}$. Then
    \begin{align*}
    \mathbf{u}\cdot\mathbf{v}&=f(\mathbf{u}+\mathbf{v})+f(\mathbf{u})+f(\mathbf{v})\\
    &=f(T(\mathbf{u})+T(\mathbf{v}))+f(T(\mathbf{u}))+f(T(\mathbf{v}))\\
    &=T(\mathbf{u})\cdot T(\mathbf{v})=0
    \end{align*}
    for every $\mathbf{v}\in\mathcal{C}$. Since $\mathcal{C}$ is an LCD code, $\mathbf{u}$ must be zero. Therefore $T$ is injective.
\end{proof}

By Theorem~\ref{LCDbasis}, any even LCD code over $\mathbb{F}_2$ has an even dimension. Thus, considering an even LCD code of dimension $2m$ in Lemma~\ref{LCDdoublyevenembedding} covers all possible cases without loss of generality.

\begin{lemma}\label{LCDdoublyevenembedding}
    Let $\mathcal{C}$ be an even LCD code over $\mathbb{F}_2$ with $\dim \mathcal{C}=2m$. Let $t$ be the number of doubly even codewords in $\mathcal{C}$ and $t_m=2^{2m-1}+(-1)^{\lfloor \frac{m+1}{2}\rfloor}2^{m-1}$, which is the number of doubly even codewords in $E_{2m+1}$. Then $n_{de}(\mathcal{C})$ is given as follows:
    \[
    n_{de}(\mathcal{C})=\begin{cases}
        n_s(\mathcal{C}) & \mbox{if}~t=t_m,\\
        n_s(\mathcal{C})+1 & \mbox{if}~t\ne t_m~\mbox{and}~m~\mbox{is even},\\
        n_s(\mathcal{C})+2 & \mbox{if}~t\ne t_m~\mbox{and}~m~\mbox{is odd}.
    \end{cases}
    \]
\end{lemma}
\begin{proof}
    We first show that $n_{de}(\mathcal{C})=n_s(\mathcal{C})$ if and only if $t=t_m$. By Theorem~\ref{shortestSO}, $n_s(\mathcal{C})=n+2m+1$ where $n$ is the length of $\mathcal{C}$. By Lemma~\ref{LCDeveniso}, $n_{de}(\mathcal{C})=n_s(\mathcal{C})$ if and only if there is a linear map $T:\mathcal{C}\to E_{2m+1}$ such that $f(T(\mathbf c))=f(\mathbf c)$ for all $\mathbf c\in\mathcal C$. Assume that $n_s(\mathcal{C})=n_{de}(\mathcal{C})$. Let $G'=[G~|~B]$ be a generator matrix of a shortest doubly even self-orthogonal embedding of $\mathcal{C}$ where $G$ is a generator matrix of $\mathcal{C}$. Since
    \[
    (G')(G')^T=GG^T+BB^T=\mathcal{O}
    \]
    and $\mathcal{C}$ is an LCD code, $BB^T$ is invertible. Then we have ${\rm{rank}}(B)=2m=\dim E_{2m+1}$. Also, since $\mathcal{C}$ and $E_{2m+1}$ are LCD codes, both $f|_\mathcal{C}$ and $f|_{E_{2m+1}}$ are nonsingular. Thus, a map $T$ given by Lemma~\ref{LCDeveniso} is an isometry between $(\mathcal{C}, f|_\mathcal{C})$ and $(E_{2m+1}, f|_{E_{2m+1}})$. Therefore, $T$ induces a bijection between the doubly even codewords of $\mathcal{C}$ and those of $E_{2m+1}$, which implies that $t=t_m$. Conversely, if $t=t_m$, then there is an isometry between $(\mathcal{C}, f|_\mathcal{C})$ and $(E_{2m+1}, f|_{E_{2m+1}})$, which implies that $n_{de}(\mathcal{C})=n_s(\mathcal{C})$.
    
    Next, we show that if $t\ne t_m$ and $m$ is odd, then $n_{de}(\mathcal{C})$ cannot be $n_s(\mathcal{C})+1$, which implies that $n_{de}(\mathcal{C})=n_s(\mathcal{C})+2$. Assume that $t\ne t_m$ and $m$ is odd. If $n_{de}(\mathcal{C})=n_s(\mathcal{C})+1=n+2m+2$, then by Lemma~\ref{LCDeveniso}, there is an injective linear map $T:\mathcal{C}\to E_{2m+2}$ such that $f(T(\mathbf c))=f(\mathbf c)$. Then $T(\mathcal{C})$, equipped with the restriction of $f|_{E_{2m+2}}$ to $T(\mathcal{C})$, is a nonsingular subspace of $E_{2m+2}$ which is isometric to $(\mathcal{C}, f|_\mathcal{C})$. Note that ${\rm Hull}(E_{2m+2})=\langle \mathbf{1}\rangle$, where $\mathbf{1}$ denotes the all-one vector of length $2m+2$. Since $T(\mathcal{C})$ is nonsingular, the inner product on $T(\mathcal{C})$ is nondegenerate. Then we have $T(\mathcal{C})\cap\langle\mathbf{1}\rangle=\mathbf{0}$, that is, $E_{2m+2}=T(\mathcal{C})\oplus\langle \mathbf{1}\rangle$. Since $f(\mathbf{1})=m+1=0$, we have 
    \[
    f(\mathbf{u}+\mathbf{1})=f(\mathbf{u})+f(\mathbf{1})+\mathbf{u}\cdot \mathbf{1}=f(\mathbf{u})
    \]
    for every $\mathbf{u}\in T(\mathcal{C})$. It follows that the number of doubly even codewords in $T(\mathcal{C})$ is half the number of doubly even codewords in $E_{2m+2}$. By Lemma~\ref{doublyevencodeword}, this number is therefore given by
    \[
    \frac{2^{2m}+(-1)^{(m+1)/2}2^{m}}{2}=2^{2m-1}+(-1)^{(m+1)/2}2^{m-1}=t_m.
    \]
    Then the number of doubly even codewords in $\mathcal{C}$ is $t_m$ which contradicts our assumption. Thus, if $t\ne t_m$ and $m$ is odd, then $n_{de}(\mathcal{C})=n_s(\mathcal{C})+2$.

    Finally, we show that if $t\ne t_m$ and $m$ is even, then $n_{de}(\mathcal{C})\le n_s(\mathcal{C})+1$, so it follows that $n_{de}(\mathcal{C})= n_s(\mathcal{C})+1$. Assume that $t\ne t_m$ and $m$ is even. Define a subset $U$ of $E_{2m+2}$ as
    \[
    U=\{(\mathbf{x}, 0)~|~\mathbf{x}\in E_{2m+1}\}.
    \]
    Clearly, $(U, f|_U)$ is isometric to $(E_{2m+1}, f|_{E_{2m+1}})$ which is nonsingular. Choose $\mathbf{x}\in U$ such that $f(\mathbf{x})=1$. Define a linear map $\varphi:U\to E_{2m+2}$ as
    \[
    \varphi(\mathbf{u})=\mathbf{u}+(\mathbf{u}\cdot \mathbf{x})\mathbf{1}
    \]
    for every $\mathbf{u}\in U$ and the all-one vector $\mathbf{1}$ of length $2m+2$. For $\mathbf{v}\in U$, if $\varphi(\mathbf{v})=0$, then $\mathbf{v}=(\mathbf{v}\cdot \mathbf{x})\mathbf{1}$. Since $\mathbf{1}\notin U$, $\mathbf{v}=0$. Thus, the dimension of $U'=\varphi(U)$ is $2m$.
    For $\mathbf{u}\in U$, since $f(\mathbf{1})=m+1=1\pmod 2$, we have
    \begin{align*}
    f(\mathbf{u}+(\mathbf{u}\cdot \mathbf{x})\mathbf{1})&=f(\mathbf{u})+f((\mathbf{u}\cdot \mathbf{x})\mathbf{1})+(\mathbf{u}\cdot \mathbf{x})(\mathbf{u}\cdot \mathbf{1})\\&=f(\mathbf{u})+(\mathbf{u}\cdot \mathbf{x})\\&=f(\mathbf{u})+f(\mathbf{u}+\mathbf{x})+f(\mathbf{u})+f(\mathbf{x})\\&=f(\mathbf{u}+\mathbf{x})+1.
    \end{align*}
    Since $\mathbf{u}\mapsto\mathbf{u}+\mathbf{x}$ is a bijection on $U$, we have that ${\rm Arf}(f|_{U'})={\rm Arf}(f|_{U})+1$. Thus, by Theorem~\ref{Arfiso}, $\mathcal{C}$ is either isometric to $U$ or $U'$. Then, by Lemma~\ref{LCDeveniso}, there is a doubly even self-orthogonal embedding of $\mathcal{C}$ with length $n+2m+2=n_s(\mathcal{C})+1$. Thus, $n_{de}(\mathcal{C})\le n_s(\mathcal{C})+1$, so it follows that $n_{de}(\mathcal{C})= n_s(\mathcal{C})+1$.
\end{proof}

\begin{theorem}\label{doublyevenev}
    Let $\mathcal{C}$ be an $[n, k]$ even code over $\mathbb{F}_2$ which is not self-orthogonal. Let 
    \[
    G=\begin{bmatrix}G({\rm Hull}(\mathcal{C}))\\A \end{bmatrix}
    \]
    be a generator matrix of $\mathcal{C}$. Let $t_A$ be the number of doubly even codewords in $\langle A\rangle$ and $t_m=2^{2m-1}+(-1)^{\lfloor \frac{m+1}{2}\rfloor}2^{m-1}$ where $m = \dim \langle A \rangle/2$. Then $n_{de}(\mathcal{C})=n_s(\mathcal{C})$ if and only if $t_A=t_m$ and ${\rm{Hull}}(\mathcal{C})$ is doubly even. If either ${\rm Hull}(\mathcal C)$ is not doubly even or $t_A\ne t_m$, then
    \[
    n_{de}(\mathcal{C})=\begin{cases}
        n_s(\mathcal{C})+1 & \mbox{if}~m~\mbox{is even},\\
        n_s(\mathcal{C})+2 & \mbox{if}~m~\mbox{is odd}.
    \end{cases}
    \]
\end{theorem}
\begin{proof}
    Since $\langle A\rangle$ is an even LCD code, by Theorem~\ref{LCDbasis}, $\dim\langle A\rangle$ is even. Let $\mathcal{C}'$ be a shortest doubly even self-orthogonal embedding of $\mathcal{C}$. 
    
    We first show that $n_{de}(\mathcal{C})=n_s(\mathcal{C})$ if and only if $t_A=t_m$ and ${\rm{Hull}}(\mathcal{C})$ is doubly even. By Proposition~\ref{deequals}, $n_{de}(\mathcal{C})=n_s(\mathcal{C})$ if and only if ${\rm{Hull}}(\mathcal{C})$ is doubly even and there is a shortest self-orthogonal embedding of $\langle A\rangle$ which is doubly even. By Lemma~\ref{LCDdoublyevenembedding}, $\langle A\rangle$ has a shortest self-orthogonal embedding which is doubly even if and only if $t_A=t_m$.

    Next, we show that if either ${\rm Hull}(\mathcal C)$ is not doubly even or $t_A\ne t_m$ and $m$ is odd, then $n_{de}(\mathcal{C})$ cannot be $n_s(\mathcal{C})+1$, so it follows that $n_{de}(\mathcal{C})=n_s(\mathcal{C})+2$. Assume that either ${\rm Hull}(\mathcal C)$ is not doubly even or $t_A\ne t_m$, and $m$ is odd. Suppose that $n_{de}(\mathcal{C})=n_s(\mathcal{C})+1$. Then there is a doubly even self-orthogonal embedding of $\langle A\rangle$ with length $n_s(\mathcal{C})+1$. By Lemma~\ref{LCDdoublyevenembedding} it is possible only if $\langle A\rangle$ has a doubly even shortest self-orthogonal embedding of length $n_s(\mathcal{C})$. Then, the number of doubly even codewords in $\langle A\rangle$ is $t_m$. 
    
    Let $G'=[G~|~B]$ be a generator matrix of $\mathcal{C}'$ which is a shortest doubly even self-orthogonal embedding of $\mathcal{C}$. Define a linear map $T:\mathcal{C}\to \langle B\rangle\subseteq E_{2m+2}$ as $T(\mathbf{x}G)=\mathbf{x}B$ for all $\mathbf{x}\in\mathbb{F}_2^{k}$. Since $\mathcal{C}'$ is doubly even, we have $f(\mathbf{c})=f(T(\mathbf{c}))$ for every $\mathbf{c}\in\mathcal{C}$. Since $\langle A\rangle$ is nonsingular and $T|_{\langle A\rangle}$ is injective, $T(\langle A\rangle)$ is a nonsingular subspace of $E_{2m+2}$. It follows that $E_{2m+2}=T(\langle A\rangle)\oplus\langle \mathbf{1}\rangle$, where $\mathbf{1}$ denotes the all-one vector of length $2m+2$. For $\mathbf{h}\in {\rm{Hull}}(\mathcal{C})$, we have
    \[
    T(\mathbf{h})\cdot T(\mathbf{u})=\mathbf{h}\cdot \mathbf{u}=0
    \]
    for every $\mathbf{u}\in\langle A\rangle$. Then $T(\mathbf{h})\in\langle\mathbf{1}\rangle$. Thus,
    \[
    f(\mathbf{h})=f(T(\mathbf{h}))=f(\mathbf{1})=m+1=0.
    \]
    Hence, ${\rm{Hull}}(\mathcal{C})$ is doubly even, which contradicts our assumption. So, $n_{de}(\mathcal{C})=n_s(\mathcal{C})+2$.

    Finally, we show that if either ${\rm Hull}(\mathcal C)$ is not doubly even or $t_A\ne t_m$ and $m$ is even, then $n_{de}(\mathcal{C})\le n_s(\mathcal{C})+1$. Hence, we must have $n_{de}(\mathcal{C})= n_s(\mathcal{C})+1$. Assume that either ${\rm Hull}(\mathcal C)$ is not doubly even or $t_A\ne t_m$, and $m$ is even. As shown in the proof of Lemma~\ref{LCDdoublyevenembedding}, there is an isometry $T:\langle A\rangle\to U\subseteq E_{2m+2}$ where $U$ is a $2m$-dimensional nonsingular subspace of $E_{2m+2}$. Define a linear map $F:\mathcal{C}={\rm Hull}(\mathcal{C})\oplus\langle A\rangle\to E_{2m+2}$ as 
    \[
    F(\mathbf{h}+\mathbf{a})=f(\mathbf{h})\mathbf{1}+T(\mathbf{a})
    \]
    for every $\mathbf{h}\in {\rm Hull}(\mathcal{C})$ and $\mathbf{a}\in\langle A\rangle$. Then we have
    \begin{align*}
    f(F(\mathbf{h}+\mathbf{a}))&=f(f(\mathbf{h})\mathbf{1}+T(\mathbf{a}))\\&=f(\mathbf{h})+f(T(\mathbf{a}))+f(\mathbf{h})(\mathbf{1}\cdot T(\mathbf{a}))\\&=f(\mathbf{h})+f(\mathbf{a})\\&=f(\mathbf{h}+\mathbf{a}).
    \end{align*}
    Therefore $n_{de}(\mathcal{C})=n+2m+2=n_s(\mathcal{C})+1$ by Lemma~\ref{LCDeveniso}.
\end{proof}

\begin{lemma}\label{LCDdoublyevenembeddingodd}
    Let $\mathcal{C}$ be an $[n, k]$ odd-like LCD code over $\mathbb{F}_2$ with a generator matrix $G$ such that $GG^T=I$. Let $m=|\{g_i~|~{\rm{wt}}(g_i)\equiv1\pmod 4\}|$ where $g_i$, $1\le i\le k$ are the rows of $G$. Then $n_{de}(\mathcal{C})$ is given as follows:
    \[
    n_{de}(\mathcal{C})=\begin{cases}
        n_s(\mathcal{C}) & \mbox{if}~m\equiv 0\pmod 4,\\
        n_s(\mathcal{C})+1 & \mbox{if}~m\equiv 3\pmod 4,\\
        n_s(\mathcal{C})+2 & \mbox{otherwise}.
    \end{cases}
    \]
\end{lemma}
\begin{proof}
    Let $G'=[G~|~B]$ be a generator matrix of a shortest doubly even self-orthogonal embedding $\mathcal{C}'$ of $\mathcal{C}$. Since $G'$ generates a doubly even code, each row $(g_i, b_i)$ has weight divisible by $4$. It follows that ${\rm{wt}}(g_i)\equiv 1\pmod 4$ if and only if ${\rm{wt}}(b_i)\equiv 3\pmod 4$, that is, $|\{b_i~|~{\rm{wt}}(b_i)\equiv3\pmod 4\}|=m$. 

    Since $BB^T=GG^T=I$, the rows of $B$ form an orthonormal basis of $\langle B\rangle$. Thus, any codeword $\mathbf{x}\in\langle B\rangle$ can be written uniquely as $\mathbf{x}=\sum_{j=1}^k\alpha_jb_j$. Since the rows of $B$ are pairwise orthogonal, the weight of $\mathbf{x}$ is given as
    \[
    {\rm{wt}}(\mathbf{x})={\rm{wt}}\left(\sum_{j=1}^k\alpha_jb_j\right)\equiv\sum_{j=1}^k\alpha_j{\rm wt}(b_j)\pmod 4.
    \]
    Thus, we have
    \begin{align*}
    \sum_{\mathbf{x}\in\langle B\rangle} i^{{\rm{wt}}(\mathbf{x})}&=\sum_{(\alpha_1,\dots,\alpha_k)\in\mathbb F_2^k}i^{\sum_{j=1}^k \alpha_j{\rm wt}(b_j)}=\prod_{j=1}^k\sum_{\alpha_j\in\mathbb{F}_2}i^{\alpha_j{\rm wt}(b_j)}\\&=\prod_{j=1}^k\left(1+i^{{\rm wt}(b_j)}\right)=(1+i)^{k-m}(1-i)^m=(-i)^m(1+i)^k,
    \end{align*}
    where $i$ denotes the imaginary unit satisfying $i^2=-1$.

    We first show that $n_{de}(\mathcal{C})=n_s(\mathcal{C})$ if and only if $m\equiv 0\pmod 4$. Suppose that $n_{de}(\mathcal{C})=n_s(\mathcal{C})$. The length of a shortest self-orthogonal embedding $\mathcal{C}'$ of $\mathcal{C}$ is $n+k$ by Theorem~\ref{shortestSO}. Therefore $\langle B\rangle=\mathbb{F}_2^k$. Since $\sum_{\mathbf{x}\in\mathbb F_2^k} i^{{\rm{wt}}(\mathbf{x})}=(1+i)^k$, we have $(-i)^m=1$, that is, $m\equiv 0\pmod 4$.
    
    Conversely, assume that $m\equiv0\pmod 4$. Define a block diagonal matrix
    \[
    B_0={\rm{diag}}\left(\underbrace{J_4-I_4,\dots,J_4-I_4}_{m/4\text{ copies}},\, I_{k-m}\right)
    \]
    over $\mathbb{F}_2$ where $J_4$ is the $4\times 4$ all one matrix. Since
    \[
    (J_4-I_4)(J_4-I_4)^T=I_4
    \]
    in $\mathbb{F}_2$, we have $B_0B_0^T=I$. Since the weight of every row of $J_4-I_4$ is $3$, $B_0$ has $m$ rows of weight $3$ and $k-m$ rows of weight $1$. Let $B'$ be the matrix obtained by permuting the rows of $B_0$ so that a row of weight $3$ is paired with each row $g_i$ satisfying ${\rm wt}(g_i)\equiv 1\pmod 4$. Note that still $(B')(B')^T=I$, and each row of $G''=[G~|~B']$ is doubly even. Hence, the code generated by $G''$ is a doubly even self-orthogonal code with length $n+k$.

    Next, we show that $n_{de}(\mathcal{C})=n_s(\mathcal{C})+1$ if and only if $m\equiv 3\pmod 4$. Assume that $n_{de}(\mathcal{C})=n_s(\mathcal{C})+1$. Since $\langle B\rangle$ is an LCD code, $\langle B\rangle^\perp$ is a one-dimensional LCD code. Let $\mathbf{h}\in\langle B\rangle^\perp$ be a nonzero element. Since $\langle B\rangle^\perp$ is an LCD code, $\mathrm{wt}(\mathbf{h})$ must be odd. Then the rows of $B$ with $\mathbf{h}$ form an orthonormal basis of $\mathbb{F}_2^{k+1}$. Then we have
    \begin{align*}
    \sum_{\mathbf{x}\in\mathbb{F}_2^{k+1}}i^{{\rm wt}(\mathbf{x})}&=\sum_{\mathbf{u}\in\langle B\rangle}\sum_{a\in \mathbb{F}_2}i^{{\rm wt}(\mathbf{u}+a\mathbf{h})}\\&=\left(\sum_{\mathbf{u}\in\langle B\rangle}i^{{\rm wt}(\mathbf{u})}\right)\left(\sum_{a\in \mathbb{F}_2}i^{{\rm wt}(a\mathbf{h})}\right)\\&=(-i)^m(1+i)^k(1+i^{{\rm wt}(\mathbf{h})}).
    \end{align*}
    Since $\sum_{\mathbf{x}\in\mathbb F_2^{k+1}} i^{{\rm{wt}}(\mathbf{x})}=(1+i)^{k+1}$, dividing by $(1+i)^k$, we have
    \[
    1+i=(-i)^m\left(1+i^{{\rm wt}(\mathbf{h})}\right).
    \]
    Since ${\rm wt}(\mathbf{h})$ is odd, it is either $1$ or $3\pmod 4$. Then $m\equiv 0\pmod 4$ or $m\equiv 3\pmod 4$. Since it cannot be congruent to $0\pmod 4$ (or we will have $n_s(\mathcal{C})=n_{de}(\mathcal{C})$), we conclude that $m\equiv 3\pmod 4$.

    Conversely, suppose that $m\equiv 3\pmod 4$. For $1\le i\le m$, let $b_i$ be the vector of length $k+1$ whose first $m+1$ coordinates are all $1$ except for the $i$-th coordinate which is $0$, and whose remaining coordinates are all $0$. For $m+1\le i\le k$, let $b_i$ be the standard $(i+1)$-th unit vector. Then $b_i$ are pairwise orthogonal and ${\rm wt}(b_i)\equiv 3\pmod 4$ if $1\le i\le m$ and ${\rm wt}(b_i)\equiv 1\pmod 4$ if $m+1\le i\le k$. Let $B$ be the matrix obtained by permuting $b_1\ldots, b_k$ so that each $b_i$ with ${\rm wt}(b_i)\equiv 3\pmod 4$ is paired with each row $g_i$ satisfying ${\rm wt}(g_i)\equiv 1\pmod 4$. Then $G'=[G~|~B]$ generates a doubly even self-orthogonal code of length $n+k+1=n_s(\mathcal{C})+1$. This completes the proof.
\end{proof}

\begin{theorem}\label{doublyevenod}
Let $\mathcal{C}$ be an $[n, k]$ odd-like code over $\mathbb{F}_2$ which is not self-orthogonal. Let
    \[
    G=\begin{bmatrix}G({\rm{Hull}}(\mathcal{C}))\\A \end{bmatrix}
    \]
    be a generator matrix of $\mathcal{C}$ where $AA^T=I$. Let $m=|\{a_i~|~{\rm{wt}}(a_i)\equiv1\pmod 4\}|$ where $a_i$ are the rows of $A$. Then $n_{de}(\mathcal{C})$ is given as follows:
    \[
    n_{de}(\mathcal{C})=\begin{cases}
        n_s(\mathcal{C}) & \mbox{if}~m\equiv 0\pmod 4~\mbox{and}~{\rm Hull}(\mathcal{C})~\mbox{is doubly even},\\
        n_s(\mathcal{C})+1 & \mbox{if}~m\equiv 3\pmod 4~\mbox{and}~{\rm Hull}(\mathcal{C})~\mbox{is doubly even},\\
        n_s(\mathcal{C})+2 & \mbox{otherwise}.
    \end{cases}
    \]
\end{theorem}
\begin{proof}
By Proposition~\ref{deequals}, $n_{de}(\mathcal{C})=n_s(\mathcal{C})$ if and only if ${\rm{Hull}}(\mathcal{C})$ is doubly even and there is a shortest self-orthogonal embedding of $\langle A\rangle$ which is doubly even, that is $m\equiv 0\pmod 4$ by Lemma~\ref{LCDdoublyevenembeddingodd}.

Next, we show that $n_{de}(\mathcal{C})=n_s(\mathcal{C})+1$ if and only if $m\equiv 3\pmod 4$ and $\mathrm{Hull}(\mathcal{C})$ is doubly even. Suppose that $n_{de}(\mathcal{C})=n_s(\mathcal{C})+1$. Let $\dim\langle A\rangle=k_a$ and let
\[
    G'=\begin{bmatrix}G({\rm Hull}(\mathcal{C})) & C\\A & B \end{bmatrix}
\]
be a generator matrix of a shortest doubly even self-orthogonal embedding $\mathcal{C}'$ of $\mathcal{C}$. Since $\langle B\rangle\subseteq\mathbb{F}_2^{k_a+1}$ is an LCD code with dimension $k_a$, it follows that $\langle B\rangle^\perp=\langle \mathbf{u}\rangle$ for some $\mathbf{u}\in\mathbb{F}_2^{k_a+1}$ such that $\mathbf{u}\cdot\mathbf{u}=1$. Since ${\rm Hull}(\mathcal{C})$ is orthogonal to $\langle A\rangle$, $\langle C\rangle$ is also orthogonal to $\langle B\rangle$, that is, $\langle C\rangle\subseteq \langle \mathbf{u}\rangle$. Also, since ${\rm Hull}(\mathcal{C})$ is self-orthogonal, $\mathbf{c}\cdot\mathbf{c}=0$ for any $\mathbf{c}\in\langle C\rangle$. Then, $\mathbf{c}$ must be $\mathbf{0}$, that is, $C=\mathcal{O}$. Hence ${\rm Hull}(\mathcal{C})$ is doubly even. Note that $[A~|~B]$ is a doubly even self-orthogonal embedding of an LCD code $\langle A\rangle$. This is possible only when $m\equiv 0\mbox{ or }3\pmod 4$ by Lemma~\ref{LCDdoublyevenembeddingodd}. Since $m$ cannot be equivalent to $0\pmod 4$ (or we will have $n_s(\mathcal{C})=n_{de}(\mathcal{C})$), we conclude that $m\equiv 3\pmod 4$.

Conversely, assume that ${\rm{Hull}}(\mathcal{C})$ is doubly even and $m\equiv 3\pmod 4$. By Lemma~\ref{LCDdoublyevenembeddingodd}, there is a doubly even self-orthogonal embedding $\langle [A~|~B]\rangle$ of $\langle A\rangle$ with length $n_s(\mathcal{C})+1$. Then, appending the zero matrix to ${\rm{Hull}}(\mathcal{C})$ gives a doubly even self-orthogonal embedding of $\mathcal{C}$ with length $n_s(\mathcal{C})+1$. This completes the proof.
\end{proof}

The $m$-th order binary extended Hamming code $\mathcal{H}_e(m)$ is the binary extended Hamming code with parameters $[2^m, 2^{m}-m-1, 4]$, which is the extended code of the binary $[2^m-1, 2^m-m-1, 3]$ Hamming code $\mathcal{H}(m)$. It is known that the $3$rd order extended Hamming code $\mathcal{H}_e(3)$ is a doubly even self-orthogonal code. Therefore, we focus on extended Hamming codes with $m\ge 4$.
\begin{theorem}\label{DoublyevenHamming}
    Let $\mathcal{H}_e(m)$ be the $m$-th order binary extended Hamming code where $m\ge 4$. Then the following holds:
\[
n_{de}(\mathcal{H}_e(m))=
\begin{cases}
2^{m+1}-2m-1, & \text{if } m\equiv0,\,3\pmod4,\\
2^{m+1}-2m, & \text{if } m\equiv1\pmod4,\\
2^{m+1}-2m+1, & \text{if } m\equiv2\pmod4.
\end{cases}
\]
\end{theorem}
\begin{proof}
    Note that $\mathcal{H}_e(m)$ has parameters $[2^m, 2^m-m-1, 4]$ and contains its dual which is the first order Reed-Muller code $\mathcal{R}(1, m)$. Let $\mathcal{H}_e(m)=\mathcal{R}(1, m)\oplus\mathcal{D}$ for some $\mathcal{D}\subseteq \mathcal{H}_e(m)$. Then, we have
    \[
    \dim\mathcal{D}=\dim\mathcal{H}_e(m)-\dim\mathcal{R}(1, m)=2^m-2m-2,
    \]
    which is even. Since $\mathcal{H}_e(m)$ is an even code, it follows from Theorem~\ref{shortestSO} that
    \[
    n_s(\mathcal{H}_e(m))=2^m+\dim\mathcal{D}+1=2^{m+1}-2m-1.
    \]
    It is known that the weight enumerator of ${\rm{Hull}}(\mathcal{H}_e(m))=\mathcal{R}(1, m)$ is given as
    \[
    W_{\mathcal{R}(1, m)}(z)=1+(2^{m+1}-2)z^{2^{m-1}}+z^{2^m}.
    \]
    Thus, if $m\ge 3$, then ${\rm{Hull}}(\mathcal{H}_e(m))$ is doubly even. Let $N$ be the number of doubly even codewords in $\mathcal{H}_e(m)$, and let
    \[
    W_{\mathcal{H}_e(m)}(x, y)=\sum_{j=0}^nA_jx^{n-j}y^j
    \]
    be the weight enumerator of $\mathcal{H}_e(m)$. Since $\mathcal{H}_e(m)$ is an even code, we have 
    \begin{align*}
    W_{\mathcal{H}_e(m)}(1, 1)+W_{\mathcal{H}_e(m)}(1, i)&=\sum_{j=0}^{n/2}A_{2j}+\sum_{j=0}^{n/2}A_{2j}i^{2j}\\&=\sum_{j=0}^{n/2} A_{2j} \left( 1 + (-1)^j \right)\\&=2(A_0+A_4+A_8+\cdots).
    \end{align*}
    Thus, $(W_{\mathcal{H}_e(m)}(1, 1)+W_{\mathcal{H}_e(m)}(1, i))/2$ counts $N$. Clearly, $W_{\mathcal{H}_e(m)}(1, 1)=|\mathcal{H}_e(m)|=2^{2^m-m-1}$. Since the weight enumerator of $\mathcal{R}(1, m)$ is given as
    \[
    W_\mathcal{R}(x, y)=x^{2^m}+(2^{m+1}-2)x^{2^{m-1}}y^{2^{m-1}}+y^{2^m},
    \]
    and $\mathcal{H}_e(m)$ is the dual of $\mathcal{R}(1, m)$, we have
    \begin{align*}
    \frac{W_{\mathcal{H}_e(m)}(1, 1)+W_{\mathcal{H}_e(m)}(1, i)}{2}&=\frac{1}{2}(2^{2^m-m-1})+\frac{1}{2|\mathcal{R}(1, m)|}W_\mathcal{R}(1+i, 1-i)\\&=2^{2^m-m-2}+\frac{1}{2^{m+2}}(2^{2^{m-1}}+(2^{m+1}-2)2^{2^{m-1}}+2^{2^{m-1}})\\&=2^{2^m-m-2}+2^{2^{m-1}-1}.
    \end{align*}
    Put $\mu=\dim\mathcal{D}/2=2^{m-1}-m-1$. Then the number of doubly even codewords $N_\mathcal{D}$ in $\mathcal{D}$ is given as
    \[
    N_\mathcal{D}=N/|\mathcal{R}(1,m)|=2^{2\mu-1}+2^{\mu-1}.
    \]
    Note that
    \[
    \left\lfloor\frac{\mu+1}{2}\right\rfloor\equiv \begin{cases}
        0\pmod 2, & \mbox{if }m\equiv 0,\, 3\pmod 4,\\
        1\pmod 2, & \mbox{if }m\equiv 1,\, 2\pmod 4.
    \end{cases}
    \]
    If $m\equiv 0,\,3\pmod 4$, then by Theorem~\ref{doublyevenev},
    \[
    n_{de}(\mathcal{H}_e(m))=2^{m+1}-2m-1.
    \]
    If $m\equiv 1\pmod 4$, then since $\mu$ is even, we have
    \[
    n_{de}(\mathcal{H}_e(m))=n_s(\mathcal{H}_e(m))+1=2^{m+1}-2m.
    \]
    On the other hand, if $m\equiv 2\pmod 4$, then since $\mu$ is odd, it follows that
    \[
    n_{de}(\mathcal{H}_e(m))=n_s(\mathcal{H}_e(m))+2=2^{m+1}-2m+1.
    \]
\end{proof}

\section{Shortest self-orthogonal embeddings of quaternary codes}
Let $\mathcal{C}$ be a linear code over $\mathbb{Z}_4$. As in the binary case, we use the notation $n_s(\mathcal{C})$ to denote the length of the shortest self-orthogonal embedding of $\mathcal{C}$. Throughout this section, for a matrix $A$ over $\mathbb{Z}_4$, we denote $A\pmod 2$ by $\overline{A}$.

\begin{proposition}\label{trivialembedding}
    Let $\mathcal{C}$ be a code of length $n$ over $\mathbb{Z}_4$. Then $n_s(\mathcal{C})\le 4n$.
\end{proposition}
\begin{proof}
    Let $G$ be a generator matrix of $\mathcal{C}$. Define $G'=[G\mid G\mid G\mid G]$. Then $(G')(G')^T=4\cdot GG^T=\mathcal{O}$. Thus $G'$ generates a self-orthogonal embedding of $\mathcal{C}$. This completes the proof.
\end{proof}

\begin{proposition}\label{Hull-type}
    Let $\mathcal{C}$ be a free code over $\mathbb{Z}_4$ with ${\rm{rank}}(\mathcal{C})=k$, and let $G$ be a generator matrix of $\mathcal{C}$. Let $\mathcal{D}$ be the code generated by the rows of $GG^T$. If the type of $\mathcal{D}$ is $4^a2^b$, then the type of ${\rm{Hull}}(\mathcal{C})$ is $4^{k-a-b}2^b$.
\end{proposition}
\begin{proof}
    Define a $\mathbb{Z}_4$-module isomorphism $\varphi:\mathbb{Z}_4^k\to \mathcal{C}$ as $\varphi(\mathbf{u})=\mathbf{u}G$ for $\mathbf{u}\in\mathbb{Z}_4^k$. Then $\varphi(\mathbf{u})\in{\rm{Hull}}(\mathcal{C})$ if and only if $\varphi(\mathbf{u})\in\mathcal{C}^\perp$, that is,
    \[
    \mathbf{0}=G\varphi(\mathbf{u})^T=G(\mathbf{u}G)^T=GG^T\mathbf{u}.
    \]
    It follows that $\mathbf{u}\in\varphi^{-1}({\rm{Hull}}(\mathcal{C}))$ if and only if $\mathbf{u}\in\mathcal{D}^\perp$. Since $\varphi$ is an isomorphism, the type of ${\rm{Hull}}(\mathcal{C})$ is equal to the type of $\varphi^{-1}({\rm{Hull}}(\mathcal{C}))=\mathcal{D}^\perp$, which is $4^{k-a-b}2^b$.
\end{proof}

\begin{lemma}\label{add-column-type}
    Let $G$ be a $k\times n$ matrix over $\mathbb{Z}_4$. Let $\mathcal{D}$ be the code generated by the rows of $GG^T$ such that the type of $\mathcal{D}$ is $4^a2^b$. For a vector $\mathbf{v}\in\mathbb{Z}_4^k$, let $G'$ be the concatenation of $G$ by $\mathbf{v}$, that is, $G'=[G\mid\mathbf{v}^T]$. Then the code $\mathcal{D}'$, generated by the rows of $G'(G')^T$, has one of the following types:
    \[
    4^{a}2^b,\,4^a2^{b\pm 1},\,4^{a+1}2^b,\,4^{a+1}2^{b-1},\,4^{a+1}2^{b-2},\,4^{a-1}2^b,\,4^{a-1}2^{b+1},\,4^{a-1}2^{b+2}.
    \]
\end{lemma}
\begin{proof}
Let the type of $\mathcal{D}'$ be $4^{a'}2^{b'}$. Since
\[
G'(G')^T=GG^T+\mathbf{v}^T\mathbf{v},
\]
and
\[
\langle \mathbf{v}^T\mathbf{v}\rangle=\mathrm{span}\{v_i\cdot\mathbf{v}\mid 1\le i\le k\},
\]
we have 
\[
\mathcal{D}'=\langle GG^T+\mathbf{v}^T\mathbf{v}\rangle\subseteq \langle GG^T\rangle+\langle\mathbf{v}^T\mathbf{v}\rangle\subseteq \mathcal{D}+\langle \mathbf{v}\rangle.
\]
Since the free rank of $\langle \mathbf{v}\rangle$ is at most one, the free rank of $\mathcal{D}+\langle\mathbf{v}\rangle$ is at most $a+1$. It follows that $a'\le a+1$. Similarly, we have $\mathcal{D}\subseteq \mathcal{D}'+\langle\mathbf{v}\rangle$ and this implies that $a\le a'+1$. Thus, we have $|a'-a|\le 1$, that is, $a'\in\{a,\, a+1,\, a-1\}$.

Since the size of a minimal generating set of $\mathcal{D}+\langle \mathbf{v}\rangle$ cannot exceed $(a+b)+1$, and $\mathbb{Z}_4$ is a principal ideal ring, we have $a'+b'\le a+b+1$. Similarly, we also have $a+b\le a'+b'+1$. Thus, we obtain $|(a'-a)+(b'-b)|\le 1$, that is, $a'+b'\in\{a+b,\,a+b+1,\,a+b-1\}$. This completes the proof.
\end{proof}

For a $k\times k$ symmetric matrix $M$ over $\mathbb{Z}_4$, define $\nu:M_{k}(\mathbb{Z}_4)\to \mathbb{Z}_{\ge 0}\cup \{\infty\}$ by
\[
\nu(M) = \begin{cases} 
0, & \text{if }M=\mathcal{O},\\
\min \{ m\in\mathbb{Z}_{> 0} \mid M + HH^T = \mathcal{O}, \, H \in M_{k \times m}(\mathbb{Z}_4) \}, & \text{if such } H \text{ exists,} \\ 
\infty, & \text{otherwise.} 
\end{cases}
\]
For a linear code $\mathcal{C}$ over $\mathbb{Z}_4$ of length $n$, let $G$ be a generator matrix of $\mathcal{C}$. It is straightforward that $n_s(\mathcal{C})-n=\nu(GG^T)$.

For matrices $A$ and $B$ over $\mathbb{Z}_4$, if there is an invertible matrix $U$ over $\mathbb{Z}_4$ such that $A=UBU^T$, then $A$ and $B$ are said to be \textit{congruent}.

\begin{remark}\label{embed-cong}
    Let $M$ be a $k\times k$ symmetric matrix over $\mathbb{Z}_4$ such that $\nu(M)=m<\infty$. Then, there is a matrix $X\in M_{k\times m}(\mathbb{Z}_4)$ such that $M+XX^T=\mathcal{O}$. For any invertible matrix $U\in GL_k(\mathbb{Z}_4)$, define $Y=UX$. Then we have
    \[
    UMU^T+YY^T=UMU^T+UXX^TU^T=U(M+XX^T)U^T=\mathcal{O}.
    \]
    Thus $\nu(UMU^T)\le m$.

    On the other hand, let $Y'$ be a $k\times \nu(UMU^T)$ matrix over $\mathbb{Z}_4$ such that $UMU^T+(Y')(Y')^T=\mathcal{O}$. Define $X'=U^{-1}Y'$. Then we have
    \begin{align*}
    M+(X')(X')^T&=U^{-1}(UMU^T)(U^T)^{-1}+U^{-1}(Y')(Y')^T(U^{-1})^T\\&=U^{-1}((UMU^T)+(Y')(Y')^T)(U^{-1})^T\\&=\mathcal{O}.
    \end{align*}
    Hence, $m\le \nu(UMU^T)$. This implies that any two congruent matrices have the same $\nu$ value.
\end{remark}

\begin{lemma}\label{embed-inequal}
    Let $M$ be a $k\times k$ symmetric matrix over $\mathbb{Z}_4$ such that 
    \[
    M=\begin{pmatrix} A & \mathcal{O}\\\mathcal{O} & B \end{pmatrix}
    \]
    for some $A\in M_{k_1}(\mathbb{Z}_4)$ and $B\in M_{k_2}(\mathbb{Z}_4)$ where $k_1+k_2=k$. Then $\nu(M)\le \nu(A)+\nu(B)$.
\end{lemma}
\begin{proof}
If either $\nu(A)$ or $\nu(B)$ is infinite, the inequality holds trivially. Suppose that both $\nu(A)$ and $\nu(B)$ are finite. Take $X\in M_{k_1\times \nu(A)}(\mathbb{Z}_4)$ and $Y\in M_{k_2\times \nu(B)}(\mathbb{Z}_4)$ such that $A+XX^T=\mathcal{O}$ and $B+YY^T=\mathcal{O}$, respectively. This gives \[
    M+\begin{pmatrix} X & \mathcal{O}\\\mathcal{O} & Y \end{pmatrix}\begin{pmatrix} X & \mathcal{O}\\\mathcal{O} & Y \end{pmatrix}^T=M-M=\mathcal{O}.
    \]
    Therefore, $\nu(M)\le \nu(A)+\nu(B)$.
\end{proof}

\begin{lemma}\label{embed-upper}
    Let $M$ be a symmetric matrix over $\mathbb{Z}_4$ such that $\nu(M)<\infty$. If the type of $\langle M\rangle$ is $4^a2^b$, then $\nu(M)\le 3a+2b$.
\end{lemma}
\begin{proof}
    If $a=b=0$, it is trivial that $\nu(M)=0$. Assume that the claim holds for every symmetric matrix $M_0$ over $\mathbb{Z}_4$ such that $\nu(M_0)<\infty$ with type $4^{a_0}2^{b_0}$ where $a_0+b_0<a+b$. We consider four cases based on the entries of $M$.
    \begin{enumerate}
        \item[(i)] Let $M$ have at least one odd entry in its diagonal. By permuting rows and columns of $M$, we assume that the $(1, 1)$-entry $u$ is odd without loss of generality. Since any odd element in $\mathbb{Z}_4$ is a unit, by applying simultaneous row and column operations, we obtain a congruent matrix
        $\begin{pmatrix}
            u & \mathbf{0}\\\mathbf{0} & M_1
        \end{pmatrix}$
        where $M_1$ is symmetric. Since $\langle M_1\rangle$ has type $4^{a-1}2^b$, by the inductive hypothesis, we have $\nu(M_1)\le 3(a-1)+2b$. Note that since $u$ is odd, we have $u=1$ or $3$ in $\mathbb Z_4$. If $u=1$, then $-u=3=1^2+1^2+1^2$, while if $u=3$, then $-u=1=1^2$ over $\mathbb Z_4$. Hence $\nu([u])\le 3$. Then, by Remark~\ref{embed-cong} and Lemma~\ref{embed-inequal}, we obtain $\nu(M)\le \nu([u])+\nu(M_1)\le 3+3(a-1)+2b=3a+2b$.
        \item[(ii)] Let every diagonal entry of $M$ be even, and let at least one off-diagonal entry be odd. By row and column permutations, we assume that the upper left $2\times 2$ submatrix of $M$ is of the form
        $M_1=\begin{pmatrix}
            m_1 & u\\u & m_2
        \end{pmatrix}$
        where $m_1$ and $m_2$ are even integers and $u$ is an odd integer. By row and column operations, we have a congruent matrix $\begin{pmatrix} M_1 & \mathcal{O}\\\mathcal{O} & M_2\end{pmatrix}$ where $M_2$ is symmetric. Define a matrix $A$ as follows:
        \[
        A=\begin{cases}
            \begin{bmatrix}
                0 & 1 & 1\\1 & 0 & -u
            \end{bmatrix}, & \text{if }m_1=m_2=2,\\[1.0em]
            \begin{bmatrix}
                0 & 1 & 1 & 1 & 1\\1 & 0 & 0 & 0 & -u
            \end{bmatrix}, & \text{if }m_1=0,\, m_2=2,\\[1.0em]
            \begin{bmatrix}
                1 & 0 & 0 & 0 & -u\\0 & 1 & 1 & 1 & 1
            \end{bmatrix}, & \text{if }m_1=2,\, m_2=0,\\[1.0em]
            \begin{bmatrix}
                0 & 1 & 1 & 1 & 1\\1 & 1 & 1 & 0 & u
            \end{bmatrix}, & \text{if }m_1=m_2=0.\\
        \end{cases}
        \]
        Since $AA^T=-M_1$, we have $\nu(M_1)\le 5\le 6$. Hence, we obtain $\nu(M)\le \nu(M_1)+\nu(M_2)\le 3\cdot 2+3(a-2)+2b=3a+2b$.
        \item[(iii)] Let every entry of $M$ be even, and let $M$ have at least one nonzero entry in its diagonal. Without loss of generality, assume that the $(1, 1)$-entry is $2$. Then we obtain a congruent matrix $\begin{pmatrix}
            2 & \mathbf{0}\\\mathbf{0}^T & M_1
        \end{pmatrix}$ where $M_1$ is symmetric. Then we have $\nu(M)\le \nu([2])+\nu(M_1)\le 2+3a+2(b-1)=3a+2b$.
        \item[(iv)] Let every entry of $M$ be even and let every diagonal entry be zero. Assume that the upper left $2\times 2$ submatrix $M_1$ of $M$ is $M_1=\begin{pmatrix}
            0 & 2\\2 & 0
        \end{pmatrix}$. Then we have a congruent matrix $\begin{pmatrix} M_1 & \mathcal{O}\\\mathcal{O} & M_2\end{pmatrix}$ where $M_2$ is symmetric. Since $-M_1=\begin{pmatrix}0 & 0 & 0 & 2\\1 & 1 & 1 &1 \end{pmatrix}\begin{pmatrix}0 & 0 & 0 & 2\\1 & 1 & 1 &1 \end{pmatrix}^T$, we have $\nu(M_1)\le 4$. It follows that $\nu(M)\le \nu(M_1)+\nu(M_2)\le 4+3a+2(b-2)=3a+2b$.  
    \end{enumerate}
    Thus, in all cases, $\nu(M)\le 3a+2b$.
\end{proof}

\begin{theorem}\label{embedding-bound}
    Let $\mathcal{C}$ be a linear code over $\mathbb{Z}_4$ with length $n$. Let $G$ be a generator matrix of $\mathcal{C}$ and let $\mathcal{D}=\langle GG^T\rangle$ be of type $4^a2^b$. Then
    \[
    n+a+b\le n_s(\mathcal{C})\le n+3a+2b.
    \]
\end{theorem}
\begin{proof}
    Let $G'$ be a generator matrix of a shortest self-orthogonal embedding $\mathcal{C}'$ of $\mathcal{C}$ such that first $n$ columns constitute $G$. By Lemma~\ref{add-column-type}, each appended column can decrease the value of $a+b$ by at most one. Therefore, at least $a+b$ columns must be appended to change the type from $4^a2^b$ to $4^02^0$. Hence $n+a+b\le n_s(\mathcal{C})$. Since the matrix $GG^T$ is symmetric, by Lemma~\ref{embed-upper}, it follows that $n_s(\mathcal{C})\le n+3a+2b$.
\end{proof}

\begin{corollary}\label{free_emb_bound}
    Let $\mathcal{C}$ be a free code over $\mathbb{Z}_4$ with length $n$ and rank $k$. If ${\rm{Hull}}(\mathcal{C})$ is of type $4^a2^b$, then 
    \[
    n+k-a\le n_s(\mathcal{C})\le n+3(k-a)-b.
    \]
\end{corollary}
\begin{proof}
By Proposition~\ref{Hull-type}, the type of $\langle GG^T\rangle$ is $4^{k-a-b}2^b$. Then, by Theorem~\ref{embedding-bound}, we have $n+k-a\le n_s(\mathcal{C})\le n+3(k-a)-b$.
\end{proof}

\begin{example}\label{ex1}
    Here, we provide two examples with the same rank and hull type: one attains the lower bound and the other attains the upper bound of Corollary~\ref{free_emb_bound}, respectively.

    Let $\mathcal{C}_1$ be the free code over $\mathbb{Z}_4$ generated by
    \[
    G_1=\begin{bmatrix}
        1 & 0 & 0 & 2 & 0 & 1\\0 & 1 & 0 & 1 & 0 & 2\\0 & 0 & 1 & 0 & 1 & 2
    \end{bmatrix}.
    \]
    Then the code $\mathcal{D}_1$ over $\mathbb{Z}_4$ generated by
    \[
    G_1G_1^T=\begin{bmatrix}
        2 & 0 & 2\\0 & 2 & 0\\2 & 0 & 2
    \end{bmatrix}
    \]
    is of type $2^2$, which implies that ${\rm Hull}(\mathcal{C}_1)$ is of type $4^12^2$. Therefore, $8\le n_s(\mathcal{C}_1)\le 10$ by Corollary~\ref{free_emb_bound}. Take $A_1=\begin{bmatrix} 1 & 1\\1 & 3\\1 & 1 \end{bmatrix}$. Since 
    \[
    A_1A_1^T=\begin{bmatrix}
        2 & 0 & 2\\0 & 2 & 0\\2 & 0 & 2
    \end{bmatrix},
    \]
    it follows that $G_1'=[G_1~|~A_1]$ generates a shortest self-orthogonal embedding of $\mathcal{C}_1$ with length $8$.

    Next, let $\mathcal{C}_2$ be the free code over $\mathbb{Z}_4$ generated by
    \[
    G_2=\begin{bmatrix}
        1 & 1 & 1 & 1 & 0 & 0 & 0 & 0 & 0 & 0\\0 & 0 & 0 & 0 & 1 & 1 & 1 & 1 & 0 & 0\\0 & 0 & 0 & 0 & 1 & 1 & 0 & 0 & 1 & 1
    \end{bmatrix}.
    \]
    Since the code $\mathcal{D}_2$ over $\mathbb{Z}_4$ generated by
    \[
    G_2G_2^T=\begin{bmatrix}
        0 & 0 & 0\\0 & 0 & 2\\0 & 2 & 0
    \end{bmatrix}
    \]
    is of type $2^2$, ${\rm Hull}(\mathcal{C}_2)$ is of type $4^12^2$. By Corollary~\ref{free_emb_bound}, $n_s(\mathcal{C}_2)$ must lie between $12$ and $14$. 
    
    Now we show that $\nu(G_2G_2^T)=\nu\left(\begin{bmatrix} 0 & 2\\2 & 0\end{bmatrix}\right)$.
    Suppose that there is a $3\times m$ matrix $X$ such that $G_2G_2^T+XX^T=\mathcal{O}$. Let $\mathbf{x}_2$ and $\mathbf{x}_3$ be the bottom two rows of $X$. Then
    \[
    \begin{bmatrix}
       0 & 2\\2 & 0
    \end{bmatrix}+
    \begin{bmatrix}
    \mathbf{x}_2\\\mathbf{x}_3
    \end{bmatrix}
    \begin{bmatrix}
    \mathbf{x}_2\\\mathbf{x}_3
    \end{bmatrix}^T=\mathcal{O}.
    \]
    Then, by Lemma~\ref{embed-inequal}, $\nu(G_2G_2^T)=\nu\left(\begin{bmatrix} 0 & 2\\2 & 0\end{bmatrix}\right)$. Define $A_2=\begin{bmatrix}
        1 & 1 & 1 & 1\\1 & 1 & 1 & 3
    \end{bmatrix}$. Since 
    \[
    A_2A_2^T=\begin{bmatrix}
       0 & 2\\2 & 0
    \end{bmatrix},
    \]
    we have $\nu(G_2G_2^T)\le 4$. Suppose that $\nu(G_2G_2^T)\le 3$. Then, there is a $2\times m$ matrix $Y$ such that $m\le 3$ and
    \[
    \begin{bmatrix}
       0 & 2\\2 & 0
    \end{bmatrix}+YY^T=\mathcal{O}.
    \]
    However, for a vector $\mathbf{y}$ of length $\le 3$ over $\mathbb{Z}_4$, in order to satisfy $\mathbf{y}\cdot\mathbf{y}=0$, the number of its odd coordinates must be zero. Thus, there is no such $Y$. Therefore, $\nu(G_2G_2^T)=4$ and $n_s(\mathcal{C}_2)=14$. Define
    \[
    B_2=\begin{bmatrix}
        0 & 0 & 0 & 0\\1 & 1 & 1 & 1\\1 & 1 & 1 & 3
    \end{bmatrix}.
    \]
    Then $G_2'=[G_2~|~B_2]$ generates a shortest self-orthogonal embedding of $\mathcal{C}_2$ of length $14$.

    These two examples illustrate that both the lower and upper bounds in Corollary~\ref{free_emb_bound} are indeed tight.
\end{example}

The following example shows
that the length of a shortest self-orthogonal embedding can also lie strictly between the bounds in Corollary~\ref{free_emb_bound}.

\begin{example}
    Let $\mathcal{C}$ be the free code over $\mathbb{Z}_4$ generated by
    \[
    G=\begin{bmatrix}
        1 & 0 & 0 & 0 & 0 & 0 & 0 & 0 & 0\\0 & 1 & 1 & 1 & 1 & 0 & 0 & 0 & 0\\0 & 1 & 1 & 1 & 3 & 0 & 0 & 0 & 0\\0 & 0 & 0 & 0 & 0 & 1 & 1 & 1 & 1
    \end{bmatrix}.
    \]
    Since
    \[
    GG^T=\begin{bmatrix}
        1 & 0 & 0 & 0\\0 & 0 & 2 & 0\\0 & 2 & 0 & 0\\0 & 0 & 0 & 0
    \end{bmatrix}
    \]
    generates a code of type $4^12^2$, ${\rm Hull}(\mathcal{C})$ is of type $4^12^2$. Then $12\le n_s(\mathcal{C})\le 16$.

    Define
    \[
    A=\begin{bmatrix}
        0 & 0 & 1 & 1 & 1\\0 & 0 & 0 & 2 & 2\\1 & 1 & 1 & 0 & 3\\0 & 0 & 0 & 0 & 0
    \end{bmatrix}.
    \]
    Since
    \[
    AA^T=\begin{bmatrix}
        3 & 0 & 0 & 0\\0 & 0 & 2 & 0\\0 & 2 & 0 & 0\\0 & 0 & 0 & 0
    \end{bmatrix}=-GG^T,
    \]
    $G'=[G~|~A]$ generates a self-orthogonal embedding of $\mathcal{C}$ with length $14$. Suppose that there is a $4\times m$ matrix $B$ over $\mathbb{Z}_4$ such that $m\le 4$ and $GG^T+BB^T=\mathcal{O}$. Let $b_i$ be the $i$-th row of $B$. Since $b_1\cdot b_1=3$, by a column permutation and multiplying columns by units if necessary, we may assume without loss of generality that $b_1=(1,1,1,0)$ or $b_1=(1,1,1,2)$. Since
    \[
    b_2\cdot b_2=b_3\cdot b_3=0,
    \]
    the number of odd coordinates in both $b_2$ and $b_3$ is either $0$ or $4$. Since
    \[
    b_1\cdot b_2=b_1\cdot b_3=0,
    \]
    it is impossible for all the coordinates to be odd, which implies that all the coordinates in both $b_2$ and $b_3$ must be even. Then $b_2\cdot b_3=0$, which is a contradiction. Therefore, $12 < n_s(\mathcal{C})=14 < 16$.
\end{example}

\begin{corollary}\label{SOembeddingres}
    Let $\mathcal{C}$ be a linear code over $\mathbb{Z}_4$. Then
    \[
    n_{de}({\rm{Res}}(\mathcal{C}))\le n_s(\mathcal{C})\le n_{de}({\rm{Res}}(\mathcal{C}))+2\dim({\rm{Tor}}(\mathcal{C})).
    \]
\end{corollary}
\begin{proof}
Let $\mathcal{C}'$ be a shortest self-orthogonal embedding of $\mathcal{C}$. For any $\mathbf{x}, \mathbf{y}\in \mathcal{C}'$, since $\mathbf{x}\cdot \mathbf{y}= 0$ in $\mathbb{Z}_4$, this implies that $\mathbf{x}\cdot \mathbf{y}\equiv 0\pmod 2$. Thus, ${\rm{Res}}(\mathcal{C'})$ is a self-orthogonal code. Also, since $\mathbf{x}\cdot\mathbf{x}\equiv 0\pmod 4$, ${\rm{Res}}(\mathcal{C'})$ is doubly even. Without loss of generality, we may assume that $\mathcal{C}'$ is constructed by appending $m=n_s(\mathcal{C})-n$ columns to a generator matrix of $\mathcal{C}$. Then, since ${\rm{Res}}(\mathcal{C})$ is obtained by puncturing the last $m$ columns of ${\rm{Res}}(\mathcal{C}')$, it follows that ${\rm{Res}}(\mathcal{C}')$ is a doubly even self-orthogonal embedding of ${\rm{Res}}(\mathcal{C})$. This gives $n_{de}({\rm{Res}}(\mathcal{C}))\le n_s(\mathcal{C})$.

Let $\mathcal{C}$ be of type $4^{k_1}2^{k_2}$ and let $G$ be a generator matrix of $\mathcal{C}$ given as
\[
 G=\begin{bmatrix}
        I_{k_1} & A & B_1+2B_2\\\mathcal{O} & 2I_{k_2} & 2D
    \end{bmatrix}
\]
where $A, B_1, B_2$ and $D$ are matrices with entries from $\{0, 1\}$. Then, we have a generator matrix $G_R'$ of a shortest doubly even self-orthogonal embedding of ${\rm Res}(\mathcal{C})$ which is given as $G_R'=[G_R~|~E]$ where $G_R=\begin{bmatrix} I_{k_1} & A & B_1 \end{bmatrix}$. Let $\mathcal{D}$ be the code generated by
\[
G'=\begin{bmatrix}
        I_{k_1} & A & B_1+2B_2 & E\\\mathcal{O} & 2I_{k_2} & 2D & \mathcal{O}
    \end{bmatrix}.
\]
Since ${\rm{Res}}(\mathcal{D})=\langle G_R'\rangle$ is self-orthogonal, every entry of $(G')(G')^T$ is divisible by $2$, which implies that the type of $\langle (G')(G')^T\rangle$ is $2^\beta$ for some $0\le \beta\le k_1+k_2=\dim({\rm Tor}(\mathcal{C}))$. Then, by Theorem~\ref{embedding-bound}, the length of a shortest self-orthogonal embedding of $\mathcal{D}$ is at most $n_{de}({\rm{Res}}(\mathcal{C}))+2\dim({\rm{Tor}}(\mathcal{C}))$, which implies that $n_s(\mathcal{C})\le n_{de}({\rm{Res}}(\mathcal{C}))+2\dim({\rm{Tor}}(\mathcal{C}))$ since $\mathcal{C}$ is a punctured code of $\mathcal{D}$.
\end{proof}

\begin{example}
    Let $\mathcal{C}$ be the code over $\mathbb{Z}_4$ generated by
    \[
    G=\begin{bmatrix}
        1 & 1 & 1 & 1\\0 & 0 & 0 & 2
    \end{bmatrix}.
    \]
    Then we have
    \[
    {\rm Res}(\mathcal{C})=\langle [1, 1, 1, 1]\rangle\quad\mbox{and}\quad{\rm Tor}(\mathcal{C})=\left\langle\begin{bmatrix}
        1 & 1 & 1 & 1\\0 & 0 & 0 & 1
    \end{bmatrix}\right\rangle.
    \]
    Then ${\rm Res}(\mathcal{C})$ is doubly even self-orthogonal. Since $GG^T=\begin{bmatrix} 0 & 2\\2 & 0 \end{bmatrix}$, as shown in Example~\ref{ex1}, $\nu(GG^T)=4=2\dim{\rm Tor}(\mathcal{C})$.

    This shows that the upper bound in Corollary~\ref{SOembeddingres} is tight. Furthermore, it will be shown later in Theorem~\ref{Preparata_emb} that the lower bound is also tight.
\end{example}

\begin{lemma}\label{matrixform}
Let $\mathcal C$ be a free code over $\mathbb Z_4$. Then there is a generator matrix $G$ of $\mathcal{C}$ which is given as
\[
    G=
    \begin{bmatrix}
        G_{h}\\
        G_{c}\\
    \end{bmatrix}
\]
where $\overline{G}=\begin{bmatrix}
        \overline{G}_{h}\\
        \overline{G}_{c}\\
    \end{bmatrix}$ generates ${\rm Res}(\mathcal{C})$, and $\overline{G}_{h}$ generates ${\rm{Hull}}({\rm Res}(\mathcal{C}))$.
\end{lemma}
\begin{proof}
    Let $\dim {\rm Res}(\mathcal{C})=k$ and $\dim {\rm{Hull}}({\rm Res}(\mathcal{C}))=\ell$. Choose a basis
    \[
    \beta=\{h_1, \ldots, h_\ell, a_1,\ldots, a_{k-\ell}\}
    \]
    of ${\rm Res}(\mathcal{C})$ such that $h_1\ldots, h_\ell$ generate ${\rm{Hull}}({\rm Res}(\mathcal{C}))$. Since every element of $\beta$ is a codeword of ${\rm Res}(\mathcal{C})$, for each $h_i$ and $a_j$, there are codewords $u_i, v_j\in\mathcal{C}$ such that $\bar{u}_i=h_i$ and $\bar{v}_j=a_j$. Suppose that
    \[
    \gamma=\sum_{i=1}^\ell \alpha_iu_i+\sum_{j=1}^{k-\ell}\beta_jv_j=0
    \]
    for some $\alpha_i, \beta_j\in\mathbb{Z}_4$. Since $\bar{\gamma}=0\pmod 2$ and $\beta$ is a basis of ${\rm Res}(\mathcal{C})$, all the $\alpha_i$ and $\beta_j$ are $0\pmod 2$, i.e., they are all even in $\mathbb{Z}_4$. Write $\alpha_i=2\alpha_i'$ and $\beta_j=2\beta_j'$. Then 
    \[
    \gamma=2\left(\sum_{i=1}^\ell \alpha_i'u_i+\sum_{j=1}^{k-\ell}\beta_j'v_j\right)=0.
    \]
    Let 
    \[
    \gamma'=\sum_{i=1}^\ell \alpha_i'u_i+\sum_{j=1}^{k-\ell}\beta_j'v_j.
    \]
    Since $2\gamma'=0$, we have $\overline{\gamma'}=0$. It follows that $\alpha_i'$ and $\beta_j'$ are all even, that is, the $\alpha_i$ and $\beta_j$ are all zero. Note that $\mathcal{C}$ is free and $\dim {\rm Res}(\mathcal{C})=k$. Thus, $u_1, \ldots, u_\ell, v_1, \ldots, v_{k-\ell}$ form a basis of $\mathcal{C}$. Let $G_h$ be the matrix whose rows are $u_i$ and let $G_c$ be the matrix whose rows are $v_j$. Then 
    \[
    G=
    \begin{bmatrix}
        G_{h}\\
        G_{c}\\
    \end{bmatrix}
    \]
    is a generator matrix of $\mathcal{C}$ such that $\overline{G}$ generates ${\rm Res}(\mathcal{C})$, and $\overline{G}_{h}$ generates ${\rm{Hull}}({\rm Res}(\mathcal{C}))$.
\end{proof}

\begin{theorem}\label{Z4embeddingres}
Let $\mathcal C$ be a code over $\mathbb Z_4$ of length $n$ with generator matrix
\[
    G=
    \begin{bmatrix}
        G_{1, h}\\
        G_{1, c}\\
        2G_2
    \end{bmatrix}
\]
where $\overline{G}_{1, h}$ generates ${\rm{Hull}}({\rm Res}(\mathcal{C}))$ and $\begin{bmatrix}
        \overline{G}_{1, h}\\
        \overline{G}_{1, c}\\
    \end{bmatrix}$ generates ${\rm Res}(\mathcal{C})$.
If
\begin{equation}\label{lifthy}
    G_{1, h}G_{1, h}^T=\mathcal{O}\quad\mbox{and}\quad 2G_{1, h}G_2^T=\mathcal{O},
\end{equation}
then $n_s(\mathcal{C})=n_{de}({\rm Res}(\mathcal{C}))$.
\end{theorem}
\begin{proof}
    By Corollary~\ref{SOembeddingres}, we have
    \[
        n_{de}({\rm Res}(\mathcal C))\le n_s(\mathcal C).
    \]
    Thus, it is enough to construct a self-orthogonal embedding of $\mathcal C$ with length $n_{de}({\rm Res}(\mathcal C))$.

    By Lemma~\ref{matrixform}, $\mathcal{C}$ always has a generator matrix of the form $G$. Since $G_{1, h}G_{1, h}^T=\mathcal{O}$, ${\rm{Hull}}({\rm Res}(\mathcal{C}))$ is doubly even. Thus, there is a shortest doubly even self-orthogonal embedding ${\rm{Res}}(\mathcal{C})'$ of ${\rm{Res}}(\mathcal{C})$ with generator matrix
    \[
    \overline{G_1}'=\begin{bmatrix}
        \overline{G}_{1, h} & \mathcal{O}\\\overline{G}_{1, c} & X_0
    \end{bmatrix}
    \]
    for some matrix $X_0$ by Proposition~\ref{doublyevengen}. Since ${\rm{Res}}(\mathcal{C})'$ is self-orthogonal, we have
    \[
    \overline{G}_{1, c}\overline{G}_{1, c}^T+X_0X_0^T=\mathcal{O}\pmod 2\quad\mbox{and}\quad\overline{G}_{1, h}\overline{G}_{1, c}^T=\mathcal{O}\pmod 2,
    \]
    which implies that
    \[
    G_{1, c}G_{1, c}^T+X_0X_0^T=2S_c\pmod 4\quad\mbox{and}\quad G_{1, h}G_{1, c}^T=2S_h\pmod 4
    \]
    for some $\{0, 1\}$-symmetric matrix $S_c$ and some $\{0, 1\}$-matrix $S_h$. Since ${\rm{Res}}(\mathcal{C})'$ is doubly even, the diagonal entries of $S_c$ are all zero. Thus, $S_c=U+U^T$ for some upper triangular matrix $U$. Since $\overline{G}_{1, c}$ generates an LCD code, $\det(X_0X_0^T)=1$, i.e., $X_0X_0^T$ is invertible. Let $Y^T=X_0^T(X_0X_0^T)^{-1}U$. Then $X_0Y^T=U$. Let $X_c=X_0+2Y$ and $X_h=S_h(X_0X_0^T)^{-1}X_0$ be matrices over $\mathbb{Z}_4$. Define
    \[
    G_1'=\begin{bmatrix}
        G_{1, h} & 2X_h\\G_{1, c} & X_c
    \end{bmatrix}.
    \]
    Since
    \begin{align*}
    G_{1, c}G_{1, c}^T+X_cX_c^T&=G_{1, c}G_{1, c}^T+X_0X_0^T+2X_0Y^T+2YX_0^T\\&=2S_c+2U+2U^T\\&=\mathcal{O},
    \end{align*}
    and
    \[
    G_{1, h}G_{1, c}^T+2X_hX_c^T=2S_h+2S_h=\mathcal{O},
    \]
    we have $(G_1')(G_1')^T=\mathcal{O}$. Thus, $G_1'$ generates a self-orthogonal code over $\mathbb{Z}_4$. 
    
    Define a binary matrix $Z$ as $Z^T=X_0^T(X_0X_0^T)^{-1}(\overline{G}_{1, c}\overline{G}_2^T)$. Since
    \[
    \overline{G}_{1, c}\overline{G}_2^T+X_0Z^T=\mathcal{O},
    \]
    as a matrix over $\mathbb{F}_2$, we have
    \[
    G_{1, c}(2G_2)^T+X_c(2Z)^T=\mathcal{O}\pmod 4.
    \]
    Let $G''$ be a matrix over $\mathbb{Z}_4$ given as
    \[
    G''=\begin{bmatrix}
        G_{1, h} & 2X_h\\G_{1, c} & X_c\\2G_2 & 2Z
    \end{bmatrix}.
    \]
    Note that $2G_{1, h}G_2^T=\mathcal{O}$. Thus $(G'')(G'')^T=\mathcal{O}$, that is, $G''$ generates a self-orthogonal embedding of $\mathcal{C}$ with length equal to $n_{de}({\rm Res}(\mathcal{C}))$. By Corollary~\ref{SOembeddingres}, the length of a shortest self-orthogonal embedding of $\mathcal{C}$ is $n_{de}({\rm Res}(\mathcal{C}))$.
\end{proof}

\begin{corollary}
    Let $\mathcal{C}$ be a linear code over $\mathbb{Z}_4$. If ${\rm{Res}}(\mathcal{C})$ is LCD, then $n_s(\mathcal{C})=n_{de}({\rm Res}(\mathcal{C}))$.
\end{corollary}
\begin{proof}
Since ${\rm Res}(\mathcal C)$ is LCD, ${\rm Hull}({\rm Res}(\mathcal C))=\{0\}$. Thus, in Theorem~\ref{Z4embeddingres}, the $G_{1,h}$ is empty and the required conditions are trivially satisfied. Therefore, the result follows from Theorem~\ref{Z4embeddingres}.
\end{proof}

The linear codes over $\mathbb{Z}_4$ whose Gray images are the binary Kerdock and Preparata codes are called the \textit{quaternary Kerdock} and \textit{quaternary Preparata codes}, respectively. It is known that quaternary Kerdock codes are self-orthogonal, and their duals are quaternary Preparata codes. As an application of Theorem~\ref{Z4embeddingres}, we determine the length of a shortest self-orthogonal embedding of the quaternary Preparata codes.

For more details on quaternary Kerdock and Preparata codes, we refer the reader to~\cite{book-HP,book-W}.

\begin{theorem}\label{Preparata_emb}
Let $m\ge 5$ be an odd integer, and let $\mathcal P_m$ be the
quaternary Preparata code of length $2^m$. Then
\[
n_s(\mathcal{P}_m)=
\begin{cases}
2^{m+1}-2m-1, & \text{if } m\equiv3\pmod4,\\
2^{m+1}-2m, & \text{if } m\equiv1\pmod4.
\end{cases}
\]
\end{theorem}
\begin{proof}
Note that the quaternary Kerdock code $\mathcal{K}_m$ of length $2^m$ is the dual of $\mathcal{P}_m$. Take a generator matrix $G$ of $\mathcal{P}_m$ which is given as
\[
G=\begin{bmatrix}
    G(\mathcal{K}_m)\\A
\end{bmatrix}
\]
where $G(\mathcal{K}_m)$ is a generator matrix of $\mathcal{K}_m$. It is known from \cite{HKCSS-1994} that both $\mathcal{K}_m$ and $\mathcal{P}_m$ are free codes, and the residue codes of those codes are given by
\[
{\rm Res}(\mathcal{P}_m)=\mathcal{H}_e(m)\quad\mbox{and}\quad{\rm Res}(\mathcal{K}_m)=\mathcal{R}(1, m),
\]
where $\mathcal{H}_e(m)$ is the binary $m$-th order extended Hamming code and $\mathcal{R}(1, m)$ is the binary first order Reed-Muller code. Since $\mathcal{R}(1, m)$ is the dual of $\mathcal{H}_e(m)$ and is self-orthogonal,
\[
{\rm{Hull}}({\rm Res}(\mathcal{P}_m))=\mathcal{R}(1, m),
\]
and therefore $\overline{G(\mathcal{K}_m)}$ generates ${\rm{Hull}}({\rm Res}(\mathcal{P}_m))$. Since the Kerdock code $\mathcal{K}_m$ is self-orthogonal, we have $G(\mathcal{K}_m)G(\mathcal{K}_m)^T=\mathcal{O}$. Here, $G$ is in the form of Theorem~\ref{Z4embeddingres}
with
\[
    G_{1,h}=G(\mathcal K_m),\quad G_{1,c}=A,
\]
and with no $G_2$. Then, by Theorem~\ref{Z4embeddingres}, $n_s(\mathcal{P}_m)$ is equal to $n_{de}(\mathcal{H}_e(m))$, which is given as
\[
\begin{cases}
2^{m+1}-2m-1, & \text{if } m\equiv3\pmod4,\\
2^{m+1}-2m, & \text{if } m\equiv1\pmod4,
\end{cases}
\]
by Theorem~\ref{DoublyevenHamming}.
\end{proof}

\section{Construction of quaternary self-orthogonal codes}

In this section, we provide an algorithm to construct self-orthogonal codes over $\mathbb{Z}_4$.

Since the exact length of the shortest self-orthogonal embedding can be precisely determined for codes satisfying the conditions of Theorem~\ref{Z4embeddingres}, we focus on these codes. For a given code $\mathcal{C}$ over $\mathbb{Z}_4$, the condition $G_{1, h}G_{1, h}^T=\mathcal{O}$ in Theorem~\ref{Z4embeddingres} implies that the hull of ${\rm Res}(\mathcal{C})$ is doubly even. To facilitate our computations, we further restrict our consideration to codes satisfying $n_s({\rm Res}(\mathcal{C}))=n_{de}({\rm Res}(\mathcal{C}))$. Under this assumption, to find a doubly even self-orthogonal embedding, we can always append the zero matrix to the hull part in the generator matrix. 

We first consider an algorithm for finding a shortest doubly even self-orthogonal embedding of a binary code $\mathcal{C}$. As shown in the proof of Lemmas~\ref{LCDeveniso},~\ref{LCDdoublyevenembedding}, and~\ref{LCDdoublyevenembeddingodd}, we can find at least one shortest doubly even self-orthogonal embedding for both even and odd-like binary codes. Propositions~\ref{evenalg} and~\ref{oddalg} serve as doubly even versions of~\cite[Theorem 4.2]{AHKL-arXiv}, providing an algorithm to construct all shortest doubly even self-orthogonal embeddings of binary codes.

\begin{proposition}\label{evenalg}
    Let $\mathcal{C}$ be an even code over $\mathbb{F}_2$ such that $n_{de}(\mathcal{C})=n_s(\mathcal{C})$ with generator matrix
    \[
    G=\begin{bmatrix}
        G({\rm Hull}(\mathcal{C}))\\A
    \end{bmatrix},
    \]
    where $\dim\langle A\rangle=2m$. Let
    \[
    G'=\begin{bmatrix}
       G({\rm Hull}(\mathcal{C})) & \mathcal{O}\\A & B
    \end{bmatrix}
    \]
    be a generator matrix of a shortest doubly even self-orthogonal embedding of $\mathcal{C}$. Define a map $q:\mathbb{F}_2^{2m}\to\mathbb{F}_2$ as
    \[
    q(\mathbf{x})=\frac{{\rm wt}(\mathbf{x}R)}{2}\pmod 2
    \]
    for all $\mathbf{x}\in\mathbb{F}_2^{2m}$, where $R=[I_{2m}~|~\mathbf{1}^T]$. Let $M$ be a unique invertible matrix over $\mathbb{F}_2$ such that $B=MR$. Then, a code $\tilde{\mathcal{C}}$ is a shortest doubly even self-orthogonal embedding of $\mathcal{C}$ if and only if $\tilde{\mathcal{C}}$ has a generator matrix
    \[
    \tilde{G}=\begin{bmatrix}
    G({\rm Hull}(\mathcal{C})) & \mathcal{O}\\A & MDR
    \end{bmatrix}
    \]
    where $D\in GL_{2m}(\mathbb{F}_2)$ such that $D(I_{2m}+J_{2m})D^T=(I_{2m}+J_{2m})$, and $q(d_i)=1$ for every row $d_i$ of $D$.
\end{proposition}
\begin{proof}
    Suppose that
    \[
    \begin{bmatrix}
       G({\rm Hull}(\mathcal{C})) & \mathcal{O}\\A & B'
    \end{bmatrix}
    \]
    generates a shortest doubly even self-orthogonal embedding of $\mathcal{C}$. As shown in the proof of Lemma~\ref{LCDeveniso}, the map $T:\langle A\rangle \to E_{2m+1}$ defined as $T(\mathbf{x}A)=\mathbf{x}B$ for every $\mathbf{x}\in\mathbb{F}_2^{2m}$ is an isomorphism. Thus, $\langle B\rangle=E_{2m+1}$, that is, the rows of $B$ form a basis of $E_{2m+1}$. Since the rows of $R$ also form a basis of $E_{2m+1}$, there is a unique $M\in GL_{2m}(\mathbb{F}_2)$ such that $B=MR$. Similarly there is a unique $M'\in GL_{2m}(\mathbb{F}_2)$ such that $B'=M'R$.

    Note that for $\mathbf{x}\in\mathbb{F}_2^{2m}$, 
    \[
    q(\mathbf{x}M)=\frac{{\rm wt}(\mathbf{x}MR)}{2}=\frac{{\rm wt}(\mathbf{x}B)}{2}=\frac{{\rm wt}(\mathbf{x}A)}{2}\pmod 2.
    \]
    Similarly, $q(\mathbf{x}M')={\rm wt}(\mathbf{x}A)/2\pmod 2$. Define $D=M^{-1}M'$. Then $B'=MDR$. Clearly $D$ is invertible. Note that $I_{2m}+J_{2m}=RR^T$. Then we have
    \[
    MRR^TM^T=BB^T=AA^T=(B')(B')^T=M'RR^T(M')^T=MDRR^TD^TM^T.
    \]
    Then we have $DRR^TD^T=RR^T$. For $\mathbf{y}\in\mathbb{F}_2^{2m}$, note that there is $\mathbf{x}\in\mathbb{F}_2^{2m}$ such that $\mathbf{y}=\mathbf{x}M$. Thus,
    \[
    q(\mathbf{y}D)=q(\mathbf{x}MD)=q(\mathbf{x}M')=q(\mathbf{x}M)=q(\mathbf{y}).
    \]
    It follows that
    \[
    q(d_i)=q(e_iD)=q(e_i)=1
    \]
    for $1\le i\le 2m$.

    Conversely, suppose that there is a matrix $D$ satisfying the given conditions. Since $RR^T=I_{2m}+J_{2m}$, we have
    \[
    d_iR\cdot d_jR=e_iDRR^TD^Te_j^T=e_iRR^Te_j^T=e_iR\cdot e_jR
    \]
    for all $i, j$. Then, for any $\mathbf{y}=\sum_{i=1}^{2m}y_ie_i\in\mathbb{F}_2^{2m}$, we have
    \begin{align*}
    q(\mathbf{y}D)&=q\left(\sum_{i=1}^{2m}y_ie_iD\right)\\&\equiv\sum_{i=1}^{2m}y_iq(d_i)+\sum_{i<j}y_iy_j(d_iR\cdot d_jR)\pmod 2\\&\equiv\sum_{i=1}^{2m}y_iq(e_i)+\sum_{i<j}y_iy_j(e_iR\cdot e_jR)\pmod 2\\&=q(\mathbf{y}).
    \end{align*}
    Define $B_D=MDR$. Then,
   \[
   \frac{{\rm wt}(\mathbf{x}B_D)}{2}=q(\mathbf{x}MD)=q(\mathbf{x}M)=\frac{{\rm wt}(\mathbf{x}B)}{2}\equiv \frac{{\rm wt}(\mathbf{x}A)}{2}\pmod 2,
   \]
   which implies that ${\rm wt}(\mathbf{x}A)+{\rm wt}(\mathbf{x}B_D)\equiv 0\pmod 4$ for any $\mathbf{x}\in\mathbb{F}_2^{2m}$. Note also that
   \[
   B_DB_D^T=MDRR^TD^TM^T=MRR^TM^T=BB^T=AA^T.
   \]
   Therefore,
   \[
    \begin{bmatrix}
    G({\rm Hull}(\mathcal{C})) & \mathcal{O}\\A & B_D
    \end{bmatrix}
    \]
    generates a shortest doubly even self-orthogonal embedding of $\mathcal{C}$.
\end{proof}

\begin{proposition}\label{oddalg}
    Let $\mathcal{C}$ be an odd-like code over $\mathbb{F}_2$ such that $n_{de}(\mathcal{C})=n_s(\mathcal{C})$ with generator matrix
    \[
    G=\begin{bmatrix}
        G({\rm Hull}(\mathcal{C}))\\A
    \end{bmatrix},
    \]
    where $\dim\langle A\rangle=m$. Let
    \[
    G'=\begin{bmatrix}
       G({\rm Hull}(\mathcal{C})) & \mathcal{O}\\A & B
    \end{bmatrix}
    \]
    be a generator matrix of a shortest doubly even self-orthogonal embedding of $\mathcal{C}$. Then, a code $\tilde{\mathcal{C}}$ is a shortest doubly even self-orthogonal embedding of $\mathcal{C}$ if and only if $\tilde{\mathcal{C}}$ has a generator matrix
    \[
    \tilde{G}=\begin{bmatrix}
    G({\rm Hull}(\mathcal{C})) & \mathcal{O}\\A & BD
    \end{bmatrix}
    \]
    where $D\in GL_{m}(\mathbb{F}_2)$ such that $DD^T=I$, and ${\rm wt}(d_i)\equiv 1\pmod 4$ for every row $d_i$ of $D$.
\end{proposition}
\begin{proof}
We first show that for $D\in GL_{m}(\mathbb{F}_2)$, $DD^T=I$, and ${\rm wt}(d_i)\equiv 1\pmod 4$ for every $i$ if and only if ${\rm wt}(\mathbf{x}D)\equiv {\rm wt}(\mathbf{x})\pmod 4$ for every $\mathbf{x}\in\mathbb{F}_2^m$.

Suppose that ${\rm wt}(\mathbf{x}D)\equiv {\rm wt}(\mathbf{x})\pmod 4$ for every $\mathbf{x}\in\mathbb{F}_2^m$. Then we have ${\rm wt}(d_i)\equiv {\rm wt}(e_i)\equiv 1\pmod 4$. Moreover, for all $\mathbf{u}, \mathbf{v}\in\mathbb{F}_2^m$,
\begin{align*}
2(\mathbf{u}D)\cdot(\mathbf{v}D)&\equiv{\rm wt}((\mathbf{u}+\mathbf{v})D)+{\rm wt}(\mathbf{u}D)+{\rm wt}(\mathbf{v}D)\pmod 4\\&\equiv {\rm wt}(\mathbf{u}+\mathbf{v})+{\rm wt}(\mathbf{u})+{\rm wt}(\mathbf{v})\pmod 4\\&\equiv 2\mathbf{u}\cdot\mathbf{v} \pmod 4
\end{align*}
Therefore, $\mathbf{u}DD^T\mathbf{v}^T=\mathbf{u}\mathbf{v}^T$, that is, $DD^T=I$.

Conversely, suppose that $D$ satisfies $DD^T=I$ and ${\rm wt}(d_i)\equiv 1\pmod 4$ for every $i$. Since $DD^T=I$, the rows $d_i$ of $D$ are pairwise orthogonal. Thus, for $\mathbf{y}=\sum_{i=1}^m y_ie_i$, we have
\[
{\rm wt}(\mathbf{y}D)={\rm wt}\left(\sum_{i=1}^m y_id_i\right)\equiv \sum_{i=1}^my_i{\rm wt}(d_i)\equiv\sum_{i=1}^my_i={\rm wt}(\mathbf{y})\pmod 4.
\]

Next, we prove the main statements. Suppose that
    \[
    \begin{bmatrix}
       G({\rm Hull}(\mathcal{C})) & \mathcal{O}\\A & B'
    \end{bmatrix}
    \]
    generates a shortest doubly even self-orthogonal embedding of $\mathcal{C}$. Since $AA^T=BB^T$ and $B$ is a square matrix by Theorem~\ref{shortestSO}, $B$ must be an invertible matrix. Similarly, $B'$ is an invertible matrix. Define $D=B^{-1}B'$. Then $B'=BD$. For any $\mathbf{x}\in\mathbb{F}_2^m$, we have
    \[
    {\rm wt}(\mathbf{x}BD)={\rm wt}(\mathbf{x}B')\equiv -{\rm wt}(\mathbf{x}A) \equiv {\rm wt}(\mathbf{x}B) \pmod 4.
    \]
    Since $B$ is invertible, it follows that ${\rm wt}(\mathbf{x}D)\equiv {\rm wt}(\mathbf{x})\pmod 4$ for every $\mathbf{x}\in\mathbb{F}_2^m$.

    Conversely, suppose that there is a matrix $D\in GL_{m}(\mathbb{F}_2)$ satisfying ${\rm wt}(\mathbf{x}D)\equiv {\rm wt}(\mathbf{x})\pmod 4$ for every $\mathbf{x}\in\mathbb{F}_2^m$. Then, for $\mathbf{x}\in\mathbb{F}_2^m$, we have
    \[
    {\rm wt}(\mathbf{x}BD)\equiv{\rm wt}(\mathbf{x}B)\equiv -{\rm wt}(\mathbf{x}A)\pmod 4.
    \]
    Then, ${\rm wt}(\mathbf{x}A)+{\rm wt}(\mathbf{x}BD)\equiv 0\pmod 4$. Since
    \[
    (BD)(BD)^T=BDD^TB^T=BB^T=AA^T,
    \]
    the code generated by
    \[
    \begin{bmatrix}
    G({\rm Hull}(\mathcal{C})) & \mathcal{O}\\A & BD
    \end{bmatrix}
    \]
    is a shortest doubly even self-orthogonal embedding of $\mathcal{C}$.
\end{proof}

When $\mathcal{C}$ is a free code over $\mathbb{Z}_4$, a shortest self-orthogonal embedding of $\mathcal{C}$ can be obtained by applying the process described in the proof of Theorem~\ref{Z4embeddingres} to a shortest doubly even self-orthogonal embedding ${\rm Res}(\mathcal{C})'$ of ${\rm Res}(\mathcal{C})$. Furthermore, based on Proposition~\ref{z4lift}, we can provide an algorithm to obtain shortest self-orthogonal embeddings of $\mathcal{C}$ arising from ${\rm Res}(\mathcal{C})'$.

\begin{proposition}\label{z4lift}
    Let $\mathcal{C}$ be a free code over $\mathbb{Z}_4$ with ${\rm rank}(\mathcal{C})=k$ and generator matrix $G$. Assume that $\mathcal{C}$ satisfies the hypotheses of Theorem~\ref{Z4embeddingres}, so that $n_{s}(\mathcal{C})=n_{de}({\rm Res}(\mathcal{C}))$. Let ${\rm Res}(\mathcal{C})'$ be a shortest doubly even self-orthogonal embedding of ${\rm Res}(\mathcal{C})$ with generator matrix $G_{de}=[\overline{G}~|~B]$ for some $k\times r$ matrix $B$ and let $\mathcal{C}'$ be a shortest self-orthogonal embedding of $\mathcal{C}$ with generator matrix $G'=[G~|~B'+2X]$ for some $\{0, 1\}$-matrix $X$ over $\mathbb{Z}_4$, where $B'$ is the $\{0,1\}$-matrix over $\mathbb Z_4$ such that $\overline{B'}=B$. Define a map $f:\mathbb{F}_2^{k\times r}\to \mathbb{F}_2^{k\times k}$ as 
    \[
    f(U)=BU^T+UB^T
    \]
    for every $U\in\mathbb{F}_2^{k\times r}$. Then a code $\tilde{\mathcal{C}}$ over $\mathbb{Z}_4$ is a shortest self-orthogonal embedding of $\mathcal{C}$ with ${\rm Res}(\tilde{\mathcal{C}})={\rm Res}(\mathcal{C})'$ if and only if $\tilde{\mathcal{C}}$ has a generator matrix
    \[
    \tilde{G}=[G~|~B'+2(X+Y)]
    \]
    where $Y$ is a $\{0,1\}$-matrix over $\mathbb Z_4$ such that $\overline{Y}\in\ker f$.
\end{proposition}
\begin{proof}
    Let $\tilde{\mathcal{C}}$ be a code over $\mathbb{Z}_4$ such that ${\rm Res}(\tilde{\mathcal{C}})={\rm Res}(\mathcal{C})'$. Then $\tilde{\mathcal{C}}$ has a generator matrix $\tilde{G}=[G~|~B'+2Z]$ for some $\{0, 1\}$-matrix $Z$ over $\mathbb{Z}_4$. Note that $\tilde{\mathcal{C}}$ is self-orthogonal if and only if \[
    GG^T+(B'+2Z)(B'+2Z)^T=GG^T+B'(B')^T+2(B'Z^T+Z(B')^T)=\mathcal{O}\pmod 4.
    \]
    Since $(\overline{G})(\overline{G})^T+BB^T=\mathcal{O}\pmod 2$, this is equivalent to
    \[
    B\overline{Z}^T+\overline{Z}B^T\equiv\frac{1}{2}(GG^T+B'(B')^T)\pmod 2.
    \]
    By subtracting
    \[
    B\overline{X}^T+\overline{X}B^T\equiv\frac{1}{2}(GG^T+B'(B')^T)\pmod 2,
    \]
    we have
    \[
    B(\overline{Z-X})^T+(\overline{Z-X})B^T=\mathcal O.
    \]
    Equivalently, there exists a $\{0,1\}$-matrix $Y$ over $\mathbb Z_4$ such that $\overline{Y}=\overline{Z+X}$ and $\overline{Y}\in\ker f$. For this $Y$, we have $2Z=2(X+Y)$ over $\mathbb Z_4$, and hence
\[
\tilde{G}=[G~|~B'+2Z]=[G~|~B'+2(X+Y)].
\]
\end{proof}

We construct free self-orthogonal codes over $\mathbb{Z}_4$ with rank $4$ to $6$ by applying our shortest self-orthogonal embedding method to codes in the BKLC(best-known linear codes) database of MAGMA. For any binary code $\mathcal{C}$ in the BKLC database with generator matrix $G$ such that $n_{de}(\mathcal{C})=n_s(\mathcal{C})$, let $G_4$ be the matrix obtained by viewing $G$ as a matrix over $\mathbb{Z}_4$. Then, the code $\mathcal{C}_4$ generated by $G_4$ satisfies the conditions of Theorem~\ref{Z4embeddingres}. Thus, we have $n_s(\mathcal{C}_4)=n_{s}(\mathcal{C})$.

Next, we find a shortest doubly even self-orthogonal embedding of $\mathcal{C}$ that maximizes the minimum Hamming distance. By applying Proposition~\ref{z4lift} to this code, we construct its lift over $\mathbb{Z}_4$, thereby finding a shortest self-orthogonal embedding of $\mathcal{C}_4$ whose minimum Lee distance is as large as possible among these lifts. The algorithm for finding such an embedding is given in Algorithm~\ref{embalg}.

\begin{algorithm}[!htb]
\caption{Search for shortest self-orthogonal embeddings over
$\mathbb Z_4$}
\label{embalg}
\begin{algorithmic}[1]
\STATE Let $\mathcal{C}$ be a free code over $\mathbb{Z}_4$ with a
generator matrix
\[
G=
\begin{bmatrix}
G_h\\
G_c
\end{bmatrix}
\]
such that $\overline{G}_h$ generates
$\operatorname{Hull}(\operatorname{Res}(\mathcal{C}))$ and
$G_hG_h^T=\mathcal{O}$ over $\mathbb{Z}_4$.
\STATE \textbf{Goal:} find a shortest self-orthogonal embedding of $\mathcal{C}$ over $\mathbb{Z}_4$ with good minimum Lee distance.
\STATE Initialize an empty list $\mathcal R$.
\STATE Construct one shortest doubly even self-orthogonal embedding $\mathrm{Res}(\mathcal{}C)'$ of $\mathrm{Res}(\mathcal{}C)$ with generator matrix $G_0'=
    \begin{bmatrix}
    G(\operatorname{Hull}(\operatorname{Res}(\mathcal C)))&\mathcal{O}\\
    A&B_0
    \end{bmatrix}$.
\IF{$\mathrm{Res}(\mathcal{}C)$ is even}
    \STATE Let $R=[I_{2m}\mid \mathbf 1^T]$ and write $B_0=CR$
    for some invertible matrix $C$ over $\mathbb F_2$.
    \STATE Let $\mathcal D_{\mathrm{even}}=\{D\in GL_{2m}(\mathbb F_2)\mid D(I_{2m}+J_{2m})D^T=I_{2m}+J_{2m},\ q(d_i)=1\text{ for all }i\}$.
    \STATE Construct all shortest doubly even
    self-orthogonal embeddings $\begin{bmatrix}
        G(\mathrm{Res}(\mathcal{C}))&\mathcal{O}\\
        A & CDR
        \end{bmatrix}$ by running over all $D\in\mathcal D_{\mathrm{even}}$.
\ELSE
    \STATE Let $\mathcal D_{\mathrm{odd}} = \{D\in GL_m(\mathbb F_2)\mid DD^T=I_m,\ \operatorname{wt}(d_i)\equiv 1\pmod 4\text{ for all }i\}$.
    \STATE Construct all shortest doubly even
    self-orthogonal embeddings $\begin{bmatrix}
        G(\mathrm{Res}(\mathcal{C}))&\mathcal{O}\\
        A & B_0D
        \end{bmatrix}$ by running over all $D\in\mathcal D_{\mathrm{odd}}$.
\ENDIF
\STATE Among these, keep only those with the highest minimum Hamming distance, up to equivalence.
\FOR{each shortest doubly even self-orthogonal embedding $[G\mid B]$}
    \STATE Find a $\{0, 1\}$-matrix $X$ over $\mathbb{Z}_4$ such that $B\overline{X}^{T}+\overline{X}B^{T} =\bigl(G(\mathcal C)G(\mathcal C)^T+B'(B')^T\bigr)/2\pmod 2$ where $B'$ is the $\{0,1\}$-matrix over $\mathbb Z_4$ such that $\overline{B'}=B$.
    \STATE Construct all lifts 
    \[
        [G(\mathcal{C})\mid B'+2(X+Y)],
    \]
    where $Y$ runs over all $\{0,1\}$-matrices over $\mathbb Z_4$ such that $f(\overline{Y})=B\overline{Y}^{T}+\overline{Y}B^{T}=\mathcal O$, and append every code that is inequivalent to all codes already in $\mathcal{R}$ to $\mathcal{R}$.
\ENDFOR
\STATE Return $\mathcal R$.
\end{algorithmic}
\end{algorithm}

\begin{example}
We consider a binary $[20, 4, 10]$ code $\mathcal{C}$ from BKLC with generator matrix
\[
G=
\begin{bmatrix}
1&0&0&0&0&1&0&1&1&0&0&1&1&0&1&1&1&1&1&1\\
0&1&1&1&1&0&1&0&0&1&1&0&0&1&0&0&1&1&1&1\\
0&1&0&0&1&1&0&1&1&1&1&0&1&0&0&0&1&0&0&1\\
0&0&1&0&1&0&0&0&1&1&0&1&1&1&0&1&0&1&0&1
\end{bmatrix}.
\]
Write
\[
G=\begin{bmatrix}
    G_h\\A
\end{bmatrix}
\]
where $G_h$ is the matrix consisting of the first two rows of $G$. Then $G_h$ generates the hull of $\mathcal{C}$. Let $G_4$ be a matrix obtained by viewing $G$ as a matrix over $\mathbb{Z}_4$, and let $\mathcal{C}_4$ be a code over $\mathbb{Z}_4$ generated by $G_4$. Then $\mathcal{C}_4$ has minimum Hamming distance $d_H(\mathcal{C}_4)=10$ and minimum Lee distance $d_L(\mathcal{C}_4)=10$.

Since $G_hG_h^T=\mathcal{O}$ in $\mathbb{Z}_4$, we can apply Theorem~\ref{Z4embeddingres}. Note that $\mathcal{C}$ is an even code. Since
\[
AA^T=\begin{bmatrix}
    0 & 1\\1 & 0
\end{bmatrix},
\]
the number of doubly even codewords in $\langle A\rangle$ is one, which is equal to the number of doubly even codewords in $E_3$ by Lemma~\ref{doublyevencodeword}. Then, by Theorem~\ref{doublyevenev}, $n_{de}(\mathcal{C})=n_s(\mathcal{C})=20+(4-2)+1=23$.

Define
\[
B_0=\begin{bmatrix}
    1 & 1 & 0\\1 & 0 & 1
\end{bmatrix}\quad\mbox{and}\quad B=\begin{bmatrix}
    0 & 0 & 0\\0 & 0 & 0\\1 & 1 & 0\\1 & 0 & 1
\end{bmatrix}.
\]
Since
\[
B_0B_0^T=\begin{bmatrix}
    0 & 1\\1 & 0
\end{bmatrix}=AA^T,
\]
$G_{de}=[G~|~B]$ generates a shortest doubly even self-orthogonal embedding of $\mathcal{C}$ with parameters $[23, 4, 12]$.

Define
\[
Z=\begin{bmatrix}
    1 & 0 & 0\\1 & 0 & 0\\0 & 0 & 0\\0 & 1 & 0
\end{bmatrix}
\]
over $\mathbb{Z}_4$. Since $Z$ satisfies
\[
B\overline{Z}^T+\overline{Z}B^T=\frac{GG^T+B'(B')^T}{2}\pmod 2,
\]
the lifted matrix
\[
\tilde{G}=[G_4~|~B'+2Z]
\]
generates a shortest self-orthogonal embedding $\tilde{\mathcal{C}}$ of $\mathcal{C}_4$ with minimum Hamming distance $d_H(\tilde{\mathcal{C}})=12$ and minimum Lee distance $d_L(\tilde{\mathcal{C}})=12$.

Define
\[
Y=\begin{bmatrix}
    0 & 0 & 0\\0 & 0 & 0\\1 & 0 & 1\\0 & 0 & 0
\end{bmatrix}
\]
over $\mathbb{Z}_4$. Then
\[
f(\overline{Y})=B\overline{Y}^T+\overline{Y}B^T=\begin{bmatrix}
    0 & 0 & 0 & 0\\0 & 0 & 0 & 0\\0 & 0 & 1 & 0\\0 & 0 & 0 & 0
\end{bmatrix}+\begin{bmatrix}
    0 & 0 & 0 & 0\\0 & 0 & 0 & 0\\0 & 0 & 1 & 0\\0 & 0 & 0 & 0
\end{bmatrix}=\mathcal{O}.
\]
Therefore, by Proposition~\ref{z4lift}, $\tilde{G}'=[G_4~|~B'+2(Z+Y)]$ generates a shortest self-orthogonal embedding $\tilde{\mathcal{C}}'$ of $\mathcal{C}_4$, and its minimum distances are $d_H(\tilde{\mathcal{C}}')=12$ and $d_L(\tilde{\mathcal{C}}')=14$.
\end{example}

\begin{example}
Let $\mathcal{C}$ be a binary $[20, 5, 9]$ code from BKLC with generator matrix
\[
G=
\begin{bmatrix}
1&1&1&1&1&1&1&1&1&1&1&1&1&1&1&1&0&0&0&0\\
0&1&0&0&0&0&1&1&0&1&0&1&1&1&1&0&1&0&0&0\\
0&0&1&0&0&1&1&0&0&1&1&0&1&0&1&1&0&1&0&0\\
0&0&0&1&0&0&0&1&1&1&1&0&0&1&1&1&0&0&1&0\\
1&1&1&1&0&0&0&1&0&1&0&0&1&0&0&1&0&0&0&1
\end{bmatrix}.
\]
Let
\[
G=\begin{bmatrix}
    G_h\\A
\end{bmatrix}
\]
where $G_h$ is the first row of $G$, which generates ${\rm Hull}(\mathcal{C})$. The lifted code $\mathcal{C}_4$ has minimum distances $d_H(\mathcal{C}_4)=d_L(\mathcal{C}_4)=9$.

Since $G_hG_h^T=\mathcal{O}$ in $\mathbb{Z}_4$, we can apply Theorem~\ref{Z4embeddingres}. Note that $\mathcal{C}$ is an odd-like code. Since $AA^T=I$ and every row of $A$ has weight congruent to $1\pmod 4$, by Theorem~\ref{doublyevenod}, $n_{de}(\mathcal{C})=n_s(\mathcal{C})=24$.

Define
\[
B_0=J_4-I_4\quad\mbox{and}\quad B=\begin{bmatrix}\mathbf{0}\\B_0\end{bmatrix}.
\]
As shown in the proof of Lemma~\ref{LCDdoublyevenembeddingodd}, $G_{de}=[G~|~B]$ generates a shortest doubly even self-orthogonal embedding of $\mathcal{C}$, and its parameters are $[24, 5, 12]$.

Define
\[
Z=\begin{bmatrix}
    0 & 0 & 0 & 0\\1 & 1 & 0 & 1\\1 & 1 & 0 & 0\\0 & 1 & 0 & 1\\0 & 0 & 1 & 0
\end{bmatrix}
\]
over $\mathbb{Z}_4$. Then the lifted matrix
\[
\tilde{G}=[G_4~|~B'+2Z]
\]
generates a shortest self-orthogonal embedding $\tilde{\mathcal{C}}$ of $\mathcal{C}_4$ with minimum distances $d_H(\tilde{\mathcal{C}})=d_L(\tilde{\mathcal{C}})=12$.

Define
\[
Y=\begin{bmatrix}
    0 & 0 & 0 & 0\\0 & 0 & 0 & 1\\1 & 0 & 1 & 0\\1 & 1 & 1 & 1\\0 & 0 & 1 & 1
\end{bmatrix}
\]
over $\mathbb{Z}_4$. Then $f(\overline{Y})=B\overline{Y}^T+\overline{Y}B^T=\mathcal{O}$. Hence, by Proposition~\ref{z4lift}, $\tilde{G}'=[G_4~|~B'+2(Z+Y)]$ generates a shortest self-orthogonal embedding $\tilde{\mathcal{C}}'$ of $\mathcal{C}_4$, and its minimum distances are $d_H(\tilde{\mathcal{C}}')=12$ and $d_L(\tilde{\mathcal{C}}')=14$.
\end{example}

For each binary code $\mathcal{C}$ in the ${\rm BKLC}(n,k)$ database with $n\le 40$, $4\le k\le 6$, and minimum distance $d\ge 3$, we construct the corresponding code $\mathcal{C}_4$ over $\mathbb{Z}_4$ and apply Algorithm~\ref{embalg} to obtain a shortest self-orthogonal embedding of $\mathcal{C}_4$. The results are summarized in Table~\ref{Z_4embeddingtable}. The purpose of this table is to show that, in many cases, adding columns to make a code self-orthogonal can also increase the minimum distance to a certain extent. We also apply Algorithm~\ref{embalg} to free codes over $\mathbb{Z}_4$ in Aydin's $\mathbb{Z}_4$ code database~\cite{AAY-web,ALO-arXiv} with $20\le n\le 128$ and $4\le k\le 6$. The corresponding results are summarized in Table~\ref{Z_4embeddingtableaydin}.

In Table~\ref{Z_4embeddingtable}, the column labeled Original code lists the parameters $[n,k,d_H,d_L]$ of the $\mathbb{Z}_4$-code $\mathcal{C}_4$ constructed from a BKLC code. In Table~\ref{Z_4embeddingtableaydin}, the column labeled Original code lists the parameters $[n,k,d_H,d_L]$ of the original free code over $\mathbb{Z}_4$ taken from Aydin's database. In both tables, the column labeled Shortest SO embedding presents the parameters $[n,k,d_H,d_L]$ of a shortest self-orthogonal embedding obtained by our algorithm with the highest minimum Lee distance. Here, $n,k,d_H,d_L$ denote the length, rank, minimum Hamming distance, and minimum Lee distance, respectively. In Table~\ref{Z_4embeddingtableaydin}, an asterisk $(*)$ indicates that the shortest self-orthogonal embedding obtained by our algorithm has the same minimum Lee distance as an existing code with the same length and rank in Aydin's $\mathbb{Z}_4$ code database. A double asterisk $(**)$ indicates that the shortest self-orthogonal embedding obtained by our algorithm has a larger minimum Lee distance than any linear code with the same length and rank listed in Aydin's $\mathbb{Z}_4$ code database. The column $d_{L, \mathrm{linear}}^{\rm Aydin}$ in Table~\ref{Z_4embeddingtableaydin} denotes the best minimum Lee distance among linear codes with the same length and rank listed in Aydin's database.

The explicit generator matrices of the codes
listed in Tables~\ref{Z_4embeddingtable} and~\ref{Z_4embeddingtableaydin}
can be found in~\cite{An-git}.

\begin{table}[htbp]
\centering
\begin{tabular}{l|l||l|l}
\toprule
Original code & Shortest SO embedding & Original code & Shortest SO embedding \\
\midrule
$[12,4,6_H,6_L]$ & $[15,4,8_H,8_L]$ & $[20,5,9_H,9_L]$ & $[24,5,12_H,14_L]$ \\
\hline
$[13,4,6_H,6_L]$ & $[16,4,8_H,8_L]$ & $[21,5,10_H,10_L]$ & $[26,5,12_H,14_L]$ \\
\hline
$[14,4,7_H,7_L]$ & $[15,4,8_H,8_L]$ & $[28,5,14_H,14_L]$ & $[31,5,16_H,16_L]$ \\
\hline
$[20,4,10_H,10_L]$ & $[23,4,12_H,14_L]$ & $[29,5,14_H,14_L]$ & $[31,5,16_H,16_L]$ \\
\hline
$[21,4,10_H,10_L]$ & $[26,4,12_H,14_L]$ & $[30,5,15_H,15_L]$ & $[31,5,16_H,16_L]$ \\
\hline
$[27,4,14_H,14_L]$ & $[30,4,16_H,16_L]$ & $[36,5,17_H,17_L]$ & $[40,5,20_H,22_L]$ \\
\hline
$[28,4,14_H,14_L]$ & $[31,4,16_H,16_L]$ & $[37,5,18_H,18_L]$ & $[42,5,20_H,22_L]$ \\
\hline
$[29,4,15_H,15_L]$ & $[30,4,16_H,16_L]$ & $[13,6,4_H,4_L]$ & $[20,6,8_H,10_L]$ \\
\hline
$[34,4,17_H,17_L]$ & $[38,4,20_H,22_L]$ & $[14,6,5_H,5_L]$ & $[18,6,8_H,8_L]$ \\
\hline
$[35,4,18_H,18_L]$ & $[40,4,20_H,22_L]$ & $[15,6,6_H,6_L]$ & $[20,6,8_H,10_L]$ \\
\hline
$[36,4,18_H,18_L]$ & $[39,4,20_H,22_L]$ & $[22,6,9_H,9_L]$ & $[26,6,12_H,12_L]$ \\
\hline
$[13,5,5_H,5_L]$ & $[17,5,8_H,8_L]$ & $[29,6,13_H,13_L]$ & $[32,6,16_H,16_L]$ \\
\hline
$[14,5,6_H,6_L]$ & $[17,5,8_H,8_L]$ & $[30,6,14_H,14_L]$ & $[32,6,16_H,16_L]$ \\
\hline
$[15,5,7_H,7_L]$ & $[16,5,8_H,8_L]$ & $[31,6,15_H,15_L]$ & $[32,6,16_H,16_L]$ \\
\bottomrule
\end{tabular}
\caption{Parameters of a shortest self-orthogonal embedding obtained by our algorithm from BKLC}
\label{Z_4embeddingtable}
\end{table}

\begin{table}[htbp]
\centering
\scriptsize
\setlength{\tabcolsep}{3pt}
\renewcommand{\arraystretch}{1.15}
\resizebox{\textwidth}{!}{%
\begin{tabular}{l|l|c||l|l|c}
\toprule
Original code & Shortest SO embedding & $d_{L, \mathrm{linear}}^{\rm Aydin}$
& Original code & Shortest SO embedding & $d_{L, \mathrm{linear}}^{\rm Aydin}$ \\
\midrule
$[20,4,10_H,16_L]$ & $[25,4,12_H,20_L]$ & $22$
& $[30,5,15_H,26_L]$ & $[31,5,16_H,28_L]^*$ & $28$ \\
\hline
$[45,4,22_H,42_L]$ & $[50,4,24_H,48_L]^*$ & $48$
& $[63,5,30_H,54_L]$ & $[66,5,32_H,54_L]^{**}$ & $36$ \\
\hline
$[54,4,26_H,51_L]$ & $[58,4,28_H,54_L]^{**}$ & $42$
& $[75,5,35_H,66_L]$ & $[76,5,36_H,68_L]^{**}$ & $42$ \\
\hline
$[56,4,27_H,54_L]$ & $[57,4,28_H,56_L]^*$ & $56$
& $[105,5,50_H,98_L]$ & $[108,5,52_H,100_L]^{**}$ & $60$ \\
\hline
$[63,4,27_H,54_L]$ & $[64,4,28_H,56_L]^{**}$ & $48$
& $[21,6,8_H,14_L]$ & $[28,6,12_H,20_L]$ & $22$ \\
\hline
$[65,4,32_H,60_L]$ & $[70,4,36_H,64_L]$ & $67$
& $[31,6,15_H,26_L]$ & $[32,6,16_H,28_L]^{**}$ & $20$ \\
\hline
$[70,4,34_H,67_L]$ & $[72,4,36_H,70_L]^*$ & $70$
& $[35,6,14_H,26_L]$ & $[42,6,16_H,32_L]$ & $34$ \\
\hline
$[80,4,39_H,77_L]$ & $[81,4,40_H,78_L]^{**}$ & $48$
& $[37,6,15_H,29_L]$ & $[42,6,16_H,32_L]$ & $34$ \\
\hline
$[91,4,39_H,78_L]$ & $[92,4,40_H,80_L]^{**}$ & $56$
& $[45,6,22_H,36_L]$ & $[48,6,24_H,36_L]^{**}$ & $27$ \\
\hline
$[95,4,46_H,90_L]$ & $[100,4,48_H,94_L]$ & $96$
& $[75,6,33_H,65_L]$ & $[81,6,36_H,70_L]^*$ & $70$ \\
\hline
$[21,5,10_H,14_L]$ & $[24,5,12_H,16_L]^{**}$ & $14$
& $[91,6,40_H,78_L]$ & $[98,6,44_H,84_L]^*$ & $84$ \\
\hline
$[22,5,10_H,16_L]$ & $[27,5,12_H,20_L]^{**}$ & $16$
& $[98,6,42_H,84_L]$ & $[105,6,48_H,90_L]$ & $94$ \\
\hline
$[26,5,12_H,20_L]$ & $[31,5,16_H,28_L]^*$ & $28$
& $[112,6,50_H,98_L]$ & $[119,6,52_H,104_L]^{**}$ & $62$ \\
\bottomrule
\end{tabular}%
}
\caption{Parameters of a shortest self-orthogonal embedding obtained by our algorithm from Aydin's database}
\label{Z_4embeddingtableaydin}
\end{table}

\section{Conclusion}
In this paper, we investigated the self-orthogonal embedding problem for linear codes over $\mathbb Z_4$. In our approach to this problem, the shortest doubly even self-orthogonal embedding problem for binary linear codes arises naturally through residue codes. We therefore completely determined the exact length of a shortest doubly even self-orthogonal embedding for any binary linear code. For codes over $\mathbb{Z}_4$, we obtained a tight bound on the number of columns required to construct a shortest self-orthogonal embedding, and showed that under specific conditions, this number coincides with the minimum number of columns needed to make its residue code a doubly even self-orthogonal code. Combining this result with our findings on the length of the shortest doubly even self-orthogonal embedding for the extended Hamming codes, we obtained the length of the shortest self-orthogonal embedding of the quaternary Preparata codes. Moreover, under the condition that a free code over $\mathbb{Z}_4$ and its residue code have the same shortest embedding length, we presented an algorithm to generate shortest self-orthogonal embeddings and provided explicit examples.

Currently, the problem of embedding a given code over $\mathbb{Z}_4$ into a self-orthogonal code with a high minimum distance remains open. Therefore, exploring such constructions will be an interesting direction for future work. Another direction for future work would be to expand the study of the self-orthogonal embedding problem to other widely studied rings in coding theory, including $\mathbb{F}_2 + u\mathbb{F}_2$.
 
\section*{Acknowledgement}
This research (J.-L. Kim) was supported in part by the BK21 FOUR (Fostering Outstanding Universities for Research) funded by the Ministry of Education (MOE, Korea), National Research Foundation of Korea (NRF) under Grant No. 4120240415042, Basic Science Research Program through the National Research Foundation of Korea (NRF) funded by the Ministry of Science and ICT under Grant No. RS-2025-24534992 and Global - Learning \& Academic research institution for Master’s·PhD students, and Postdocs(LAMP) Program of the National Research Foundation of Korea(NRF) grant funded by the Ministry of Education(No. RS-2024-00441954).

This research (S. Ling) was supported by Nanyang Technological University under Research Grant 04INS000047C230GRT01. Part of this work was done when S. Ling was visiting Sogang University. He thanks the institution for the hospitality.


\begin{thebibliography}{00}

\bibitem{An-git} J. An, SO embedding over $\mathbb{Z}_4$ datasets, Available from: \url{https://github.com/junmin518-droid/SO-embedding-over-Z4-datasets.git}. 2026.

\bibitem{AHKL-arXiv} J. An, J.-H. Hong, J.-L. Kim and H. Lim, ``Shortest LCD embeddings of binary, ternary and quaternary linear codes," to appear in Adv. Math. Commun., arXiv:2601.20600, 2026.

\bibitem{AKKLW-arXiv} J. An, N. Kaplan, J.-L. Kim, J. Luo and G. Wang, ``Shortest self-orthogonal embeddings of binary linear codes," arXiv:2511.05440, 2025.

\bibitem{A-1941} C. Arf, ``Untersuchungen \"uber quadratische formen in k\"orpern der charakteristik 2. Teil I.," J. Reine Angew. Math., vol. 183, pp. 148-167, 1941.

\bibitem{AAY-web} N. Aydin, T. Asamov and B. Yoshino, Online database of $\mathbb{Z}_4$ codes, Available from: \url{http://quantumcodes.info/Z4/}, 2021.

\bibitem{ALO-arXiv} N. Aydin, Y. Lu and V. R. Onta, ``An updated database of $\mathbb{Z}_4$ codes," arXiv:2208.06832, 2022.

\bibitem{BSBM-1997} A. Bonnecaze, P. S\'ole, C. Bachoc and B. Mourrain, ``Type II codes over $\mathbb{Z}_4$," IEEE Trans. Inform. Theory, vol. 43, no. 3, pp. 969-976, 1997.

\bibitem{BSC-1995} A. Bonnecaze, P. S\'ole and A. R. Calderbank, ``Quaternary quadratic residue codes and unimodular lattices," IEEE Trans. Inform. Theory, vol. 41, no. 2, pp. 366-377, 1995.

\bibitem{CMKH-1996}
A. R. Calderbank, G. McGuire, P. V. Kumar, and T. Helleseth,
``Cyclic codes over $\mathbb Z_4$, locator polynomials, and Newton's identities," IEEE Trans. Inform. Theory, vol. 42, no. 1, pp. 217-226, 1996.

\bibitem{CS-1995} A. R. Calderbank and N. J. A. Sloane, ``Modular and $p$-adic cyclic codes," Des. Codes Cryptogr., vol. 6, pp. 21-35, 1995.

\bibitem{CMTQ-2018-1} C. Carlet, S. Mesnager, C. Tang and Y. Qi, ``New characterization and parametrization of LCD codes," IEEE Trans. Inform. Theory, vol. 65, no. 1, pp. 39-49, 2019.

\bibitem{CMTQP-2018-2} C. Carlet, S. Mesnager, C. Tang, Y. Qi, and R. Pellikaan. ``Linear codes over $\mathbb{F}_q$ are equivalent to LCD codes for $q>3$," IEEE Trans. Inform. Theory, vol. 64, no. 4, pp. 3010-3017, 2018.

\bibitem{CS-1993} J. H. Conway and N. J. A. Sloane, ``Self-dual codes over the integers modulo $4$," J. Combin. Theory Ser. A, vol. 62, pp. 30-45, 1993.

\bibitem{DHS-2001} S. T. Dougherty, M. Harada and P. S\'ole, ``Shadow codes over $\mathbb{Z}_4$," Finite Fields Appl., vol. 7, pp. 507--529, 2001.

\bibitem{FGPL-1998} J. Fields, P. Gaborit, J. S. Leon, and V. Pless, ``All self-dual $\mathbb Z_4$ codes of length 15 or less are known," IEEE Trans. Inform. Theory, vol. 44, no. 1, pp. 311-322, 1998.

\bibitem{FST-1992} G. D. Forney Jr, N. J. A. Sloane and M. D. Trott, ``The Nordstrom-Robinson code is the binary image of the octacode," Coding and Quantization: DIMACS/IEEE workshop, Amer. Math. Soc., 1992.

\bibitem{G-1996} P. Gaborit, ``Mass formulas for self-dual codes over $\mathbb Z_4$ and $\mathbb F_q+u\mathbb F_q$ rings," IEEE Trans. Inform. Theory, vol. 42, no. 4, pp. 1222-1228, 1996.

\bibitem{HKCSS-1994} A. R. Hammons Jr, P. V. Kumar, A. R. Calderbank, N. J. A. Sloane and P. S\'ole, ``The $\mathbb{Z}_4$-linearity of Kerdock, Preparata, Goethals and related codes," IEEE Trans. Inform. Theory, vol. 40, no. 2, pp. 301-319, 1994.

\bibitem{H-2010} M. Harada, ``Extremal Type II $\mathbb Z_4$-codes of lengths 56 and 64," J. Combin. Theory Ser. A, vol. 117, no. 8, pp. 1285-1288, 2010.

\bibitem{H-2012} M. Harada, ``Optimal self-dual $\mathbb Z_4$-codes and a unimodular lattice in dimension 41," Finite Fields Appl., vol. 18, no. 3, pp. 529-536, 2012.

\bibitem{HK-2000} M. Harada and M. Kitazume, ``$\mathbb Z_4$-code constructions for the Niemeier lattices and their embeddings in the Leech lattice," European J. Combin., vol. 21, no. 4, pp. 473-485, 2000.

\bibitem{book-HP} W. C. Huffman and V. Pless, {\it Fundamentals of Error-Correcting Codes}, Cambridge University Press, Cambridge, 2003.

\bibitem{KC-2022} J.-L. Kim and W.-H. Choi, ``Self-orthogonality matrix and Reed-Muller codes," IEEE Trans. Inform. Theory, vol. 68, no. 11, pp. 7159-7164, 2022.

\bibitem{KKL-2021} J.-L. Kim, Y.-H. Kim and N. Lee, ``Embedding linear codes into self-orthogonal codes and their optimal minimum distances," IEEE Trans. Inform. Theory, vol. 67, no. 6, pp. 3701-3707, 2021.

\bibitem{KLL-2026} J.-L. Kim, H. Liu and J. Luo, ``How to expand a self-orthogonal code," Adv. Math. Commun., vol. 22, pp. 233-239, 2026. 

\bibitem{LZ-2019} C. Li and P. Zeng, ``Constructions of linear codes with one-dimensional hull," IEEE Trans. Inform. Theory, vol. 65, no. 3, pp. 1668-1676, 2019.

\bibitem{book-MS} F. J. MacWilliams and N. J. A. Sloane, {\it The Theory of Error-Correcting Codes}. North-Holland, London, 1977.

\bibitem{M-2024} A. Munemasa and R. A. L. Betty, ``Classification of extremal Type II $\mathbb Z_4$-codes of length 24," Designs, Codes and Cryptography, vol. 92, pp. 771-785, 2024.

\bibitem{PQ-1996} V. Pless and Z. Qian, ``Cyclic codes and quadratic residue codes over $\mathbb Z_4$," IEEE Trans. Inform. Theory, vol. 42, no. 5, pp. 1594--1600, 1996.

\bibitem{R-1999} E. M. Rains, ``Optimal self-dual codes over $\mathbb{Z}_4$," Discrete Math., vol. 203, pp. 215-228, 1999.

\bibitem{R-2000} E. M. Rains, ``Bounds for self-dual codes over $\mathbb{Z}_4$," J. Combin. Theory Ser. A, vol. 90, pp. 195-202, 2000.

\bibitem{ST-2026} M. Shi and S. Tao, ``Expanding self-orthogonal codes over a ring $\mathbb{Z}_4$ to self-dual codes and unimodular lattices," Cryptogr. Commun., 2026. https://doi.org/10.1007/s12095-026-00884-0

\bibitem{book-W} Z. Wan, {\it Quaternary Codes}. World Scientific, New Jersey, 1997.

\bibitem{WL-arXiv} J. Wang and J. Luo, ``Shortest embeddings of linear codes with arbitrary hull dimension," arXiv:2604.08843, 2026.

\end{thebibliography}
\end{document}